\documentclass[12pt, onecolumn]{IEEEtran}
\renewcommand{\baselinestretch}{1.5}

\usepackage{graphicx}

\newtheorem{theorem}{\bf Theorem}
\newtheorem{corollary}{\bf Corollary}

\newtheorem{lemma}{\bf Lemma}

\begin{document}
%
\title{Channel Coding in Random Access Communication over Compound Channels}
%
%
\author{Zheng Wang,
Jie Luo
\thanks{The authors are with the Electrical and Computer Engineering Department, Colorado State University, Fort Collins, CO 80523. E-mail: \{zhwang, rockey\}@engr.colostate.edu. }
\thanks{This work was supported by the National Science Foundation under Grants CCF-1016985 and CNS-1116134. Any opinions, findings, and conclusions or recommendations expressed in this paper are those of the authors and do not necessarily reflect the views of the National Science Foundation.}
}

\maketitle

\begin{abstract}

Due to the short and bursty incoming messages, channel access activities in a wireless random access system are often fractional. The lack of frequent data support consequently makes it difficult for the receiver to estimate and track the time varying channel states with high precision. This paper investigates random multiple access communication over a compound wireless channel where channel realization is known neither at the transmitters nor at the receiver. An achievable rate and error probability tradeoff bound is derived under the non-asymptotic assumption of a finite codeword length. The results are then extended to the random multiple access system where the receiver is only interested in decoding messages from a user subset.
\end{abstract}

\begin{keywords}
channel coding, compound channel, finite codeword length, random access
\end{keywords}

%
\IEEEpeerreviewmaketitle


\section{Introduction}
\label{SectionIntroduction}
In random multiple access communication, users (transmitters) determine their communication rates individually, without sharing the rate information either among each other or with the receiver \cite{ref Luo09}. With the absence of rate coordination among users, reliable message recovery is not always possible \cite{ref Massey85}. The receiver in this case decodes the transmitted messages if a pre-determined error probability requirement can be satisfied, or reports a collision otherwise \cite{ref Luo09}.

Information theoretic channel coding in time-slotted random multiple access communication over a discrete-time memoryless channel was recently investigated in \cite{ref Luo09}\cite{ref WangJournal}. Assume that channel coding is applied only within each time slot (or packet). It was shown in \cite{ref Luo09} that the fundamental performance limitation of the system can be characterized using an achievable rate region in the following sense. Asymptotically as the codeword length (or time slot length) is taken to infinity, the receiver is able to recover the messages reliably if the communication rate vector (which includes the rates of all users) happens to be inside the rate region, and to reliably report a collision if the rate vector happens to be outside the region \cite{ref Luo09}. The achievable rate region was shown to coincide with the Shannon information rate region without a convex hull operation \cite{ref Luo09}. In \cite{ref WangJournal}, the asymptotic coding result was further strengthened to a rate and error probability tradeoff bound under the assumption of a finite codeword length. A bound on the achievable error exponent was obtained consequently \cite{ref WangJournal}.

Both \cite{ref Luo09} and \cite{ref WangJournal} assumed that the channel state information is known at the receiver. Unfortunately, since random access communication deals with bursty short messages, transmission activities of a user are often fractional. Without frequent data support, accurate real-time channel estimation and tracking become difficult at the receiver. Understanding the system performance limitation without channel state information therefore becomes essential \cite{ref Lapidoth98}. In this paper, we illustrate how coding theorems developed in \cite{ref Luo09}\cite{ref WangJournal} can be extended to random multiple access communication over a compound discrete-time memoryless channel \cite{ref Blackwell59}\cite{ref Wolfowitz59}, consisting of a family (set) of channels over which the communication could take place. Both the transmitters and the receiver know about the compound channel set, but not the actual channel realization. As in \cite{ref Luo09}\cite{ref WangJournal}, we assume that time is partitioned into slots of equal length, and we focus on channel coding within one time slot. We define the communication rate of a user as the normalized number of information nats encoded in a packet (or a time slot).

The compound channel communication problem investigated in this paper is different from a conventional one in the following two key aspects. First, in a conventional system, information rates are jointly determined by the transmitters and the receiver \cite{ref Cover05}, while communication rates in a random access system are determined distributively and the rate information is unknown at the receiver \cite{ref Luo09}. Second, in a conventional system, in order to achieve reliable communication, the transmitted rate vector should be supported by all channel realizations in the compound set \cite{ref Csiszar81}\cite{ref Lapidoth98}. In random access communication, however, even though the receiver needs to guarantee the reliability of its decoding output, the receiver also has the additional choice of reporting a collision to avoid confusing the upper layer networking \cite{ref Bertsekas92}. This therefore allows the transmitted rate vector to be supported only by a subset of channel realizations. If the actual channel realization belongs to this subset, the receiver should decode the messages. Otherwise, the receiver should report a collision. Clearly, the decoding and collision report decisions made at the receiver are affected jointly by the communication rates of the users and the actual channel realization.

To address these key aspects in the system model, we assume that the receiver chooses an ``operation region", which is a set of rate vector and channel realization pairs. If the transmitted rate vector and channel realization pair is within the operation region, the receiver {\em intends} to decode the messages, otherwise the receiver {\em intends} to report a collision (or outage). We define the decoding error probability and the collision miss detection probability similarly to \cite{ref WangJournal}, and define the system error probability as the maximum of the two. An upper bound on the achievable system error probability is derived under the assumption of a finite codeword length. We then show how the compound channel results help in obtaining error performance bounds for the random multiple access system where the receiver is only interested in recovering messages from a user {\em subset} \cite{ref Luo09}. This is based on the observation that, conditioned on the receiver not decoding messages for the rest of the users, the impact of their communication activities on the user subset of interest is equivalent to that of a compound channel.

\section{Multiple Random Access Communication over A Compound Channel}
\label{Section MRAC}
Consider a $K$-user time-slotted random access system over a compound discrete-time memoryless channel. Time is slotted with each slot equaling $N$ symbol durations, which is also the length of a packet or a codeword. We assume that channel coding is only applied within each time slot or packet. The compound channel consists of a family of discrete-time memoryless channels, characterized by a set of conditional probabilities $\left\{ P_{Y|X_1,\cdots,X_K}^{(1)}, \cdots, P_{Y|X_1,\cdots,X_K}^{(H)} \right\}$ with cardinality $H$, where, for $k \in \{1,\cdots, K\}$, $X_k \in \mathcal X$ is the channel input symbol of user $k$ with $\mathcal X$ being the finite input alphabet, and $Y \in \mathcal Y$ is the channel output symbol with $\mathcal Y$ being the finite output alphabet. In each time slot, a channel realization is randomly generated from this set and remains static throughout the slot duration. We assume that all users and the receiver know the compound channel set, but not the actual channel realization. For the time being, we will assume that $H<\infty$. The case when the compound channel set contains an infinite number of channels will be discussed at the end of this section.

Assume that at the beginning of a time slot, according to the message availability and the MAC layer protocol, each user, say user $k$ ($k\in\{1, \cdots, K\}$) chooses an arbitrary communication rate $r_k \in \{r_{k1},\cdots, r_{kM}\}$ in nats per symbol, where $\{r_{k1},\cdots, r_{kM}\}$ is a pre-determined finite rate set of user $k$ with cardinality $M$. Neither the other users nor the receiver knows the actual rate realization for each transmission, although they are shared with the rate set information. The user then encodes $\lfloor Nr_k \rfloor$ number of data nats, denoted by a message $w_k$, into a packet (codeword) with $N$ symbols, using a random coding scheme specified as in \cite{ref Luo09}\cite{ref Shamai07} and also in the following. For all $k\in\{1,\cdots,K\}$, we assume that user $k$ is equipped with a codebook library $\mathcal{L}_k = \{\mathcal{C}_{k\theta_k}: \theta_k \in \Theta_k\}$ in which codebooks are indexed by a set $\Theta_k$. Each codebook has $M$ classes of codewords. The $i^{th}$ ($i \in \{1,\cdots,M\}$) codeword class has $\lfloor e^{Nr_{ki}} \rfloor$ codewords with the same length of $N$ symbols. In contrast to a conventional coding scheme, here each codeword in the codebook corresponds to a message and rate pair $(w_k,r_k)$ \cite{ref Luo09}\cite{ref WangJournal}. Let $\mathcal{C}_{k\theta_k}(w_k,r_k)_j$ be the $j^{th}$ symbol of the codeword corresponding to message and rate pair $(w_k,r_k)$ in codebook $\mathcal{C}_{k\theta_k}$. User $k$ first selects the codebook by generating $\theta_k$ according to a distribution $\vartheta_k$ such that the random variables $X_{(w_k,r_k),j}: \theta_k \rightarrow \mathcal{C}_{k\theta_k}(w_k,r_k)_j$ are i.i.d. according to an input distribution $P_{X|r_k}$\footnote{The input distribution is assumed to be a function of the communication rate. In other words, different communication rates may correspond to different input distributions.}. The codebook $\mathcal{C}_{k\theta_k}$ is then used to map $(w_k,r_k)$ into a codeword, denoted by $\mbox{\boldmath$x$}_{(w_k,r_k)}$. After encoding, the codewords of all users are sent to the receiver over the compound channel.

To simplify the notation, we use bold font variable to denote the corresponding variables of all users. For example, $\mbox{\boldmath$w$}$ and $\mbox{\boldmath$r$}$ denote the messages and communication rates of all users. $\mbox{\boldmath$P$}_{\mbox{\scriptsize\boldmath$X$}|\mbox{\scriptsize\boldmath$r$}}$ denote the input distributions of all users, etc. Given a vector variable, say $\mbox{\boldmath$r$}$, we use $r_i$ to denote its element corresponding to user $i$. Let $\mathcal S \subset \{1,\cdots,K\}$ be a user subset, and $\bar{\mathcal{S}}$ be its complement. We use $\mbox{\boldmath$r$}_{\mathcal S}$ to denote the vector that is extracted from $\mbox{\boldmath$r$}$ with only elements corresponding to users in $\mathcal S$. By using the vector notation of the channel input symbols, the compound channel set is also denoted by $\{P^{(1)}_{Y|\mbox{\scriptsize\boldmath$X$}}, \cdots, P^{(H)}_{Y|\mbox{\scriptsize\boldmath$X$}}\}$.

We assume that the receiver is shared with the random codebook generation algorithms and hence knows the randomly generated codebooks of all users. Before packet transmission, the receiver pre-determines an ``operation region'' $\mathcal R =\{(\mbox{\boldmath$r$},P_{Y|\mbox{\scriptsize\boldmath$X$}})\}$, which is a set of rate vector and channel realization pair, where each entry of $\mbox{\boldmath$r$}$ is chosen from the corresponding rate set, i.e., $r_k \in \{r_{k1},\cdots,r_{kM}\}$ ($k\in\{1,\cdots,K\}$), and $P_{Y|\mbox{\scriptsize\boldmath$X$}} \in \left\{ P_{Y|\mbox{\scriptsize\boldmath$X$}}^{(1)}, \cdots, P_{Y|\mbox{\scriptsize\boldmath$X$}}^{(H)} \right\}$. Let $(\mbox{\boldmath$r$},P_{Y|\mbox{\scriptsize\boldmath$X$}})$ be the actual realization of the transmitted rate vector and channel pair. We assume that the receiver {\it intends} to decode all messages if $(\mbox{\boldmath$r$},P_{Y|\mbox{\scriptsize\boldmath$X$}}) \in \mathcal R$, and {\it intends} to report a collision if $(\mbox{\boldmath$r$},P_{Y|\mbox{\scriptsize\boldmath$X$}}) \notin \mathcal R$. Note that the actual rate and channel realization $(\mbox{\boldmath$r$},P_{Y|\mbox{\scriptsize\boldmath$X$}})$ is unknown at the receiver. Therefore the receiver needs to make decisions whether to decode messages or to report a collision only based on the received channel symbols. More specifically, in each time slot, upon receiving the channel output symbols $\mbox{\boldmath$y$}$, the receiver estimates the rate and channel pair, denoted by $(\hat{\mbox{\boldmath$r$}},\hat{P}_{Y|\mbox{\scriptsize\boldmath$X$}})$, for all users. The receiver outputs the corresponding estimated message and rate vector pair $(\hat{\mbox{\boldmath$w$}},\hat{\mbox{\boldmath$r$}})$ if $(\hat{\mbox{\boldmath$r$}},\hat{P}_{Y|\mbox{\scriptsize\boldmath$X$}}) \in \mathcal R$ and a pre-determined decoding error probability requirement is satisfied. Otherwise, the receiver reports a collision. Also note that, whether the receiver should recover the messages or report a collision not only depends on the rates, but also depends on the channel realization. In other words, for the same transmission rate vector, the receiver may be designed to take different actions for different channel realizations. This is opposed to the conventional compound channel communication scenario where, if a rate is supported by the system, the receiver should always decode the messages irrespective of the channel realization.

Given the operation region $\mathcal R$, and conditioned on that $(\mbox{\boldmath$w$},\mbox{\boldmath$r$})$ is transmitted over channel $P_{Y|\mbox{\scriptsize\boldmath$X$}}$, we define the following three error probabilities. The decoding error probability, for $(\mbox{\boldmath$w$},\mbox{\boldmath$r$},P_{Y|\mbox{\scriptsize\boldmath$X$}}) $ with $(\mbox{\boldmath$r$},P_{Y|\mbox{\scriptsize\boldmath$X$}}) \in \mathcal R$, is defined as
\begin{equation}\label{MC-DecodingErrorDef}
P_{e(\mbox{\scriptsize\boldmath$w$},\mbox{\scriptsize\boldmath$r$},P_{Y|\mbox{\tiny\boldmath$X$}})} = Pr \left\{ (\hat{\mbox{\boldmath$w$}},\hat{\mbox{\boldmath$r$}}) \neq (\mbox{\boldmath$w$},\mbox{\boldmath$r$})|(\mbox{\boldmath$w$},\mbox{\boldmath$r$},P_{Y|\mbox{\scriptsize\boldmath$X$}}) \right\}, \quad
\forall (\mbox{\boldmath$w$},\mbox{\boldmath$r$},P_{Y|\mbox{\scriptsize\boldmath$X$}}) , (\mbox{\boldmath$r$},P_{Y|\mbox{\scriptsize\boldmath$X$}}) \in \mathcal R.
\end{equation}
The collision miss detection probability, for $(\mbox{\boldmath$w$},\mbox{\boldmath$r$},P_{Y|\mbox{\scriptsize\boldmath$X$}})$ with $(\mbox{\boldmath$r$},P_{Y|\mbox{\scriptsize\boldmath$X$}}) \notin \mathcal R$, is defined as
\begin{eqnarray}\label{MC-CollisionMissDef}
&& \bar{P}_{c(\mbox{\scriptsize\boldmath$w$},\mbox{\scriptsize\boldmath$r$},P_{Y|\mbox{\tiny\boldmath$X$}})} = 1 - Pr \left\{ \mbox{``collision''}|(\mbox{\boldmath$w$},\mbox{\boldmath$r$},P_{Y|\mbox{\scriptsize\boldmath$X$}}) \right\} - Pr\left\{ (\hat{\mbox{\boldmath$w$}},\hat{\mbox{\boldmath$r$}}) = (\mbox{\boldmath$w$},\mbox{\boldmath$r$})|(\mbox{\boldmath$w$},\mbox{\boldmath$r$},P_{Y|\mbox{\scriptsize\boldmath$X$}}) \right\},\nonumber\\
&& \qquad \qquad \qquad \qquad \qquad \qquad \qquad \qquad \qquad \qquad \qquad \quad \forall (\mbox{\boldmath$w$},\mbox{\boldmath$r$},P_{Y|\mbox{\scriptsize\boldmath$X$}}) , (\mbox{\boldmath$r$},P_{Y|\mbox{\scriptsize\boldmath$X$}}) \notin \mathcal R.
\end{eqnarray}
Note that in (\ref{MC-CollisionMissDef}), when $(\mbox{\boldmath$r$},P_{Y|\mbox{\scriptsize\boldmath$X$}}) \notin \mathcal R$, we have excluded the correct message and rate pair estimation from the collision miss detection event.

Let $\mathcal S \subset \{1,\cdots,K\}$ be an arbitrary user subset. Assume that $\sum_{k\notin\mathcal S} r_k \le I_{(\mbox{\scriptsize\boldmath$r$},P_{Y|\mbox{\tiny\boldmath$X$}})} (\mbox{\boldmath$X$}_{\bar{\mathcal S}};\mbox{\boldmath$Y$}|\mbox{\boldmath$X$}_{\mathcal S})$ for all $ (\mbox{\boldmath$r$},P_{Y|\mbox{\scriptsize\boldmath$X$}}) \in \mathcal R$, where $\mbox{\boldmath$X$}_{\mathcal S}$ denotes the channel input symbols of users in set $\mathcal S$, and $\mbox{\boldmath$X$}_{\bar{\mathcal S}}$ denotes the channel input symbols of users not in set $\mathcal S$. $I_{(\mbox{\scriptsize\boldmath$r$},P_{Y|\mbox{\tiny\boldmath$X$}})}$ is the mutual information function computed using input distribution corresponding to rate vector $\mbox{\boldmath$r$}$ (i.e., $\mbox{\boldmath$P$}_{\mbox{\scriptsize\boldmath$X$}|\mbox{\scriptsize\boldmath$r$}}$) and channel $P_{Y|\mbox{\scriptsize\boldmath$X$}}$. We define the system error probability $P_{es}$ as
\begin{eqnarray}\label{MC-SystemErrorDef}
P_{es} = \max \left\{ \max_{(\mbox{\scriptsize\boldmath$w$},\mbox{\scriptsize\boldmath$r$},P_{Y|\mbox{\tiny\boldmath$X$}}),(\mbox{\scriptsize\boldmath$r$},P_{Y|\mbox{\tiny\boldmath$X$}}) \in \mathcal R} P_{e(\mbox{\scriptsize\boldmath$w$},\mbox{\scriptsize\boldmath$r$},P_{Y|\mbox{\tiny\boldmath$X$}})}, \max_{(\mbox{\scriptsize\boldmath$w$},\mbox{\scriptsize\boldmath$r$},P_{Y|\mbox{\tiny\boldmath$X$}}),(\mbox{\scriptsize\boldmath$r$},P_{Y|\mbox{\tiny\boldmath$X$}})\notin \mathcal R} \bar{P}_{c(\mbox{\scriptsize\boldmath$w$},\mbox{\scriptsize\boldmath$r$},P_{Y|\mbox{\tiny\boldmath$X$}})}\right\}.
\end{eqnarray}

The following theorem gives an upper bound on the achievable system error probability $P_{es}$.

\begin{theorem}\label{TheoremMC}
Consider $K$-user multiple random access communication over a compound discrete-time memoryless channel  $\left\{ P^{(1)}_{Y|\mbox{\scriptsize\boldmath$X$}}, \cdots, P^{(H)}_{Y|\mbox{\scriptsize\boldmath$X$} }\right\}$, where $H<\infty$ is a positive integer. Let $\mbox{\boldmath$P$}_{\mbox{\scriptsize\boldmath$X$}|\mbox{\scriptsize\boldmath$r$}}$ be the input distribution for all users and all rates. Let $\mathcal R$ be the operation region. Assume finite codeword length $N$. There exists a decoding algorithm, whose system error probability $P_{es}$ is upper bounded by,
\begin{eqnarray}\label{SystemErrorTMC}
P_{es}\le \max\left\{ \max_{(\mbox{\scriptsize \boldmath $r$},P_{Y|\mbox{\tiny\boldmath$X$}})\in \mathcal{R} }\sum_{\mathcal{S}\subset\{1, \cdots, K\}} \left[ \begin{array}{l} \sum_{(\tilde{\mbox{\scriptsize \boldmath $r$}},\tilde{P}_{Y|\mbox{\tiny\boldmath$X$}})\in \mathcal{R}, \tilde{\mbox{\scriptsize \boldmath $r$}}_{\mathcal{S}}= \mbox{\scriptsize \boldmath $r$}_{\mathcal{S}} } \exp\{-N E_m(\mathcal{S}, \mbox{\boldmath$r$},\tilde{\mbox{\boldmath $r$}}, P_{Y|\mbox{\scriptsize\boldmath$X$}}, \tilde{P}_{Y|\mbox{\scriptsize\boldmath$X$}} )\}  \\   +  \max_{(\mbox{\scriptsize \boldmath $r$}', P'_{Y|\mbox{\tiny\boldmath$X$}})\not\in \mathcal{R}, \mbox{\scriptsize \boldmath $r$}'_{\mathcal{S}}= \mbox{\scriptsize \boldmath $r$}_{\mathcal{S}} } \exp\{-NE_i(\mathcal{S}, \mbox{\boldmath$r$},\mbox{\boldmath $r$}', P_{Y|\mbox{\scriptsize\boldmath$X$}}, P'_{Y|\mbox{\scriptsize\boldmath$X$}} ) \} \end{array}  \right],\right. \nonumber\\
\left.\max_{(\tilde{\mbox{\scriptsize \boldmath $r$}},\tilde{P}_{Y|\mbox{\tiny\boldmath$X$}})\not\in \mathcal{R}}\sum_{\mathcal{S}\subset\{1, \cdots, K\}} \sum_{(\mbox{\scriptsize \boldmath $r$},P_{Y|\mbox{\tiny\boldmath$X$}})\in \mathcal{R}, \mbox{\scriptsize \boldmath $r$}_{\mathcal{S}}=\tilde{\mbox{\scriptsize \boldmath $r$}}_{\mathcal{S}}} \max_{(\mbox{\scriptsize \boldmath $r$}',P'_{Y|\mbox{\tiny\boldmath$X$}})\not\in \mathcal{R}, \mbox{\scriptsize \boldmath $r$}'_{\mathcal{S}}= \tilde{\mbox{\scriptsize \boldmath $r$}}_{\mathcal{S}} } \exp\{ -N E_i (\mathcal{S}, \mbox{\boldmath$r$},\mbox{\boldmath $r$}', P_{Y|\mbox{\scriptsize\boldmath$X$}}, P'_{Y|\mbox{\scriptsize\boldmath$X$}} ) \}      \right\},
\end{eqnarray}
where $E_m(\mathcal{S}, \mbox{\boldmath$r$},\tilde{\mbox{\boldmath $r$}}, P_{Y|\mbox{\scriptsize\boldmath$X$}}, \tilde{P}_{Y|\mbox{\scriptsize\boldmath$X$}} )$ and $E_i(\mathcal{S}, \mbox{\boldmath$r$},\mbox{\boldmath $r$}', P_{Y|\mbox{\scriptsize\boldmath$X$}}, P'_{Y|\mbox{\scriptsize\boldmath$X$}} )$ are given by
\begin{eqnarray}\label{EmEiMultiMC}
&& E_m(\mathcal{S}, \mbox{\boldmath$r$},\tilde{\mbox{\boldmath $r$}}, P_{Y|\mbox{\scriptsize\boldmath$X$}}, \tilde{P}_{Y|\mbox{\scriptsize\boldmath$X$}} ) = \max_{0<\rho \le 1} -\rho \sum_{k\not\in \mathcal{S}}\tilde{r}_k + \max_{0<s\le 1} -\log \sum_Y \sum_{\mbox{\scriptsize \boldmath $X$}_{\mathcal{S}}} \prod_{k\in \mathcal{S}} P_{X|r_k}(X_k)                    \nonumber \\
&& \quad \times \left(\sum_{\mbox{\scriptsize \boldmath $X$}_{\bar{\mathcal{S}}}}\prod_{k \not\in \mathcal{S}}P_{X|r_k}(X_k)P_{Y|\mbox{\scriptsize\boldmath$X$}}(Y|\mbox{\boldmath $X$})^{1-s}\right)  \left(\sum_{\mbox{\scriptsize \boldmath $X$}_{\bar{\mathcal{S}}}}\prod_{k \not\in \mathcal{S}}P_{X|\tilde{r}_k}(X_k)\tilde{P}_{Y|\mbox{\scriptsize\boldmath$X$}}(Y|\mbox{\boldmath $X$})^{\frac{s}{\rho}} \right)^{\rho},                       \nonumber \\
&& E_i(\mathcal{S}, \mbox{\boldmath$r$},\mbox{\boldmath $r$}', P_{Y|\mbox{\scriptsize\boldmath$X$}}, P'_{Y|\mbox{\scriptsize\boldmath$X$}} ) = \max_{0<\rho \le 1} -\rho \sum_{k\not\in \mathcal{S}}r_k + \max_{0<s \le 1-\rho} - \log \sum_Y \sum_{\mbox{\scriptsize \boldmath $X$}_{\mathcal{S}}} \prod_{k\in \mathcal{S}} P_{X|r_k}(X_k)               \nonumber \\
&& \quad \times \left(\sum_{\mbox{\scriptsize \boldmath $X$}_{\bar{\mathcal{S}}}}\prod_{k \not\in \mathcal{S}}P_{X|r_k}(X_k)P_{Y|\mbox{\scriptsize\boldmath$X$}}(Y|\mbox{\boldmath $X$})^{\frac{s}{s+\rho}} \right)^{s+\rho}\left(\sum_{\mbox{\scriptsize \boldmath $X$}_{\bar{\mathcal{S}}}}\prod_{k \not\in \mathcal{S}}P_{X|r'_k}(X_k)P'_{Y|\mbox{\scriptsize\boldmath$X$}}(Y|\mbox{\boldmath $X$})\right)^{1-s}.
\end{eqnarray}
$\QED$
\end{theorem}
The proof of Theorem \ref{TheoremMC} is given in Appendix \ref{AppendixTheoremMC}.

When the compound channel is randomly generated at the beginning but remains static afterwards, one can take codeword length to infinity to obtain the system error exponent as $E_s = \lim_{N\rightarrow\infty} -\frac{1}{N} \log P_{es}$. The following lower bound on the achievable system error exponent $E_{s}$ can be easily derived from Theorem \ref{TheoremMC}.
\begin{corollary}\label{CorollaryMC}
The system error exponent of a $K$-user multiple random access system over compound discrete-time memoryless channels given in Theorem \ref{TheoremMC} is lower-bounded by
\begin{eqnarray}\label{SystemErrorExponentMC}
E_{s} \ge \min\left\{\min_{\mathcal{S}\subset \{1,\cdots,K\}} \min_{ (\mbox{\scriptsize\boldmath$r$},P_{Y|\mbox{\tiny\boldmath$X$}} ), (\tilde{\mbox{\scriptsize\boldmath$r$}},\tilde{P}_{Y|\mbox{\tiny\boldmath$X$}} ) \in \mathcal R, }  E_m(\mathcal{S}, \mbox{\boldmath$r$},\tilde{\mbox{\boldmath $r$}}, P_{Y|\mbox{\scriptsize\boldmath$X$}}, \tilde{P}_{Y|\mbox{\scriptsize\boldmath$X$}} ),\right. \nonumber\\
\left.\min_{\mathcal{S}\subset \{1,\cdots,K\}}  \min_{ (\mbox{\scriptsize\boldmath$r$},P_{Y|\mbox{\tiny\boldmath$X$}} ) \in \mathcal R, (\tilde{\mbox{\scriptsize\boldmath$r$}},\tilde{P}_{Y|\mbox{\tiny\boldmath$X$}} ) \notin \mathcal R, }E_i(\mathcal{S}, \mbox{\boldmath$r$},\tilde{\mbox{\boldmath $r$}}, P_{Y|\mbox{\scriptsize\boldmath$X$}}, \tilde{P}_{Y|\mbox{\scriptsize\boldmath$X$}} ) \right\},
\end{eqnarray}
where $E_m(\mathcal{S}, \mbox{\boldmath$r$},\tilde{\mbox{\boldmath $r$}}, P_{Y|\mbox{\scriptsize\boldmath$X$}}, \tilde{P}_{Y|\mbox{\scriptsize\boldmath$X$}} )$ and $E_i(\mathcal{S}, \mbox{\boldmath$r$},\tilde{\mbox{\boldmath $r$}}, P_{Y|\mbox{\scriptsize\boldmath$X$}}, \tilde{P}_{Y|\mbox{\scriptsize\boldmath$X$}} )$ are given in (\ref{EmEiMultiMC}).
\end{corollary}

Compared with the error exponent derived in \cite[Corollary 2]{ref WangJournal}, we can see that, even though the channel stays static forever, the system still needs to pay a penalty in error exponent performance for not knowing the channel at the receiver\footnote{We assume that such a conclusion should be well known for the conventional compound channel communication. However, we are not able to find a reference that made such a clear statement.}.

In both Theorem \ref{TheoremMC} and Corollary \ref{CorollaryMC}, we have assumed that there are only a finite number of channels in the compound set. Next, we will extend the result to the case when the cardinality of the compound channel set can be infinity.

We first assume that the the channels in the compound set can be partitioned into $H$ classes, denoted by $\left\{ \mathcal{F}^{(1)},\cdots, \mathcal{F}^{(H)}\right\}$, where $H < \infty$ is a positive integer. For example, if the compound channel set contains fading channels with continuous channel gains, one could quantize the channel gains and define the set of channels with the same quantization outcome as one channel class. We next assume that the receiver should choose an operation region $\mathcal{R}$ to satisfy the following constraint for any rate vector $\mbox{\boldmath $r$}$ and channel class $\mathcal{F}\in \left\{ \mathcal{F}^{(1)},\cdots, \mathcal{F}^{(H)}\right\}$.
\begin{equation}
\mbox{C1: For any } (\mbox{\boldmath $r$}, \mathcal{F}), \mbox{ either } (\mbox{\boldmath $r$}, P_{Y|\mbox{\scriptsize\boldmath$X$}}) \in \mathcal{R} \mbox{ } \forall P_{Y|\mbox{\scriptsize\boldmath$X$}} \in \mathcal{F}, \mbox{ or } (\mbox{\boldmath $r$}, P_{Y|\mbox{\scriptsize\boldmath$X$}}) \not\in \mathcal{R} \mbox{ } \forall P_{Y|\mbox{\scriptsize\boldmath$X$}} \in \mathcal{F}.
\label{ConstraintC1}
\end{equation}
We say $(\mbox{\boldmath $r$}, \mathcal{F})\in \mathcal{R}$ if $(\mbox{\boldmath $r$}, P_{Y|\mbox{\scriptsize\boldmath$X$}}) \in \mathcal{R}$ for all $P_{Y|\mbox{\scriptsize\boldmath$X$}}\in \mathcal{F}$, and we say  $(\mbox{\boldmath $r$}, \mathcal{F})\not\in \mathcal{R}$ otherwise.

For each channel class $\mathcal{F}$ and for each channel output symbol $Y$ and input symbol vector $\mbox{\boldmath$X$}$, we define the following upper and lower bounds on the conditional probability values yielded by channels in $\mathcal{F}$, denoted by $P^{\mathcal{F}}_{\max}(Y|\mbox{\boldmath$X$})$ and $P^{\mathcal{F}}_{\min}(Y|\mbox{\boldmath$X$})$,
\begin{eqnarray}
\label{UpperLowerBounds}
&& P^{\mathcal{F}}_{\max}(Y|\mbox{\boldmath$X$}) = \max_{P_{Y|\mbox{\tiny\boldmath$X$}} \in \mathcal{F}} P_{Y|\mbox{\scriptsize\boldmath$X$}}(Y|\mbox{\boldmath$X$}) , \qquad P^{\mathcal{F}}_{\min}(Y|\mbox{\boldmath$X$}) = \min_{P_{Y|\mbox{\tiny\boldmath$X$}} \in \mathcal{F}} P_{Y|\mbox{\scriptsize\boldmath$X$}}(Y|\mbox{\boldmath$X$}).
\end{eqnarray}
The following theorem gives an upper bound on the achievable system error probability.

\begin{theorem}
\label{TheoremContinuousChannel}
Consider a $K$-user multiple random access communication system over a compound discrete-time memoryless channel. Assume that the compound set is partitioned into $H$ classes, denoted by $\left\{ \mathcal{F}^{(1)},\cdots, \mathcal{F}^{(H)}\right\}$, where $H$ is a finite positive integer. Assume that the operation region $\mathcal{R}$ satisfies constraint C1 given in (\ref{ConstraintC1}). The system error probability $P_{es}$ is upper bounded as follows.
\begin{eqnarray}
\label{PesCC}
 P_{es} & \le & \max \left\{  \max_{(\mbox{\scriptsize\boldmath$r$},\mathcal{F}) \in \mathcal{R}} \sum_{\mathcal{S} \subset \{1,\cdots,K\}} \left[ \max_{(\mbox{\scriptsize\boldmath$r$}', \mathcal{F}') \notin \mathcal{R},\mbox{\scriptsize\boldmath$r$}'_{\mathcal{S}} = \mbox{\scriptsize\boldmath$r$}_{\mathcal{S}}} \exp \left\{ -NE_i(\mathcal{S}, \mbox{\boldmath$r$},\mbox{\boldmath $r$}', \mathcal{F}, \mathcal{F}' ) \right\} \right. \right.\nonumber\\
&& \qquad \left. + \sum_{(\tilde{\mbox{\scriptsize\boldmath$r$}},\tilde{\mathcal{F}})  \in \mathcal{R},\tilde{\mbox{\scriptsize\boldmath$r$}}_{\mathcal{S}} = \mbox{\scriptsize\boldmath$r$}_{\mathcal{S}}} \exp \left\{ -NE_m(\mathcal{S}, \mbox{\boldmath$r$},\tilde{\mbox{\boldmath $r$}}, \mathcal{F}, \tilde{\mathcal{F}} ) \right\} \right] , \nonumber\\
&& \quad \left.\max_{(\tilde{\mbox{\scriptsize\boldmath$r$}},\tilde{\mathcal{F}})  \notin \mathcal{R}} \sum_{\mathcal{S} \subset \{1,\cdots,K\}}
\left[ \sum_{(\mbox{\scriptsize\boldmath$r$},\mathcal{F}) \in \mathcal{R},\mbox{\scriptsize\boldmath$r$}_{\mathcal{S}} = \tilde{\mbox{\scriptsize\boldmath$r$}}_{\mathcal{S}}} \max_{(\mbox{\scriptsize\boldmath$r$}', \mathcal{F}') \notin \mathcal{R},\mbox{\scriptsize\boldmath$r$}'_{\mathcal{S}} = \mbox{\scriptsize\boldmath$r$}_{\mathcal{S}}} \exp \left\{ -NE_i(\mathcal{S}, \mbox{\boldmath$r$},\mbox{\boldmath $r$}', \mathcal{F}, \mathcal{F}' ) \right\} \right] \right\}.
\end{eqnarray}
where $E_m(\mathcal{S}, \mbox{\boldmath$r$},\tilde{\mbox{\boldmath $r$}}, \mathcal{F}, \tilde{\mathcal{F}} )$ and $E_i(\mathcal{S}, \mbox{\boldmath$r$},\mbox{\boldmath $r$}', \mathcal{F}, \mathcal{F}' )$ are given by
\begin{eqnarray}
\label{EmEiCC}
&& E_m(\mathcal{S}, \mbox{\boldmath$r$},\tilde{\mbox{\boldmath $r$}}, \mathcal{F}, \tilde{\mathcal{F}} ) = \max_{0<\rho \le 1} -\rho \sum_{k\not\in \mathcal{S}}\tilde{r}_k + \max_{0<s\le 1} -\log \sum_Y \sum_{\mbox{\scriptsize \boldmath $X$}_{\mathcal{S}}} \prod_{k\in \mathcal{S}} P_{X|r_k}(X_k)                    \nonumber \\
&& \quad \times \left(\sum_{\mbox{\scriptsize \boldmath $X$}_{\bar{\mathcal{S}}}}\prod_{k \not\in \mathcal{S}}P_{X|r_k}(X_k)P^{\mathcal{F}}_{\max}(Y|\mbox{\boldmath $X$})P^{\mathcal{F}}_{\min}(Y|\mbox{\boldmath $X$})^{-s}\right)  \left(\sum_{\mbox{\scriptsize \boldmath $X$}_{\bar{\mathcal{S}}}}\prod_{k \not\in \mathcal{S}}P_{X|\tilde{r}_k}(X_k) P^{\tilde{\mathcal{F}}}_{\max}(Y|\mbox{\boldmath $X$})^{\frac{s}{\rho}} \right)^{\rho}.                       \nonumber \\
&& E_i(\mathcal{S}, \mbox{\boldmath$r$},\mbox{\boldmath $r$}', \mathcal{F}, \mathcal{F}' ) = \max_{0<\rho \le 1} -\rho \sum_{k\not\in \mathcal{S}}r_k + \max_{0<s \le 1-\rho} - \log \sum_Y \sum_{\mbox{\scriptsize \boldmath $X$}_{\mathcal{S}}} \prod_{k\in \mathcal{S}} P_{X|r_k}(X_k)               \nonumber \\
&& \quad \times \left(\sum_{\mbox{\scriptsize \boldmath $X$}_{\bar{\mathcal{S}}}}\prod_{k \not\in \mathcal{S}}P_{X|r_k}(X_k)P^{\mathcal{F}}_{\max}(Y|\mbox{\boldmath $X$})P^{\mathcal{F}}_{\min}(Y|\mbox{\boldmath $X$})^{\frac{-\rho}{s+\rho}} \right)^{s+\rho}\left(\sum_{\mbox{\scriptsize \boldmath $X$}_{\bar{\mathcal{S}}}}\prod_{k \not\in \mathcal{S}}P_{X|r'_k}(X_k)P^{\mathcal{F}'}_{\max}(Y|\mbox{\boldmath $X$})\right)^{1-s}.     \nonumber\\
\end{eqnarray}
$\QED$
\end{theorem}

The proof of Theorem \ref{TheoremContinuousChannel} is given in Appendix \ref{AppendixTheoremContinuousChannel}. As shown in the proof that, in order to make decoding and collision report decisions, the receiver only needs to search over the finite number of channel classes using statistics $P^\mathcal{F}_{\max}$ and $P^\mathcal{F}_{\min}$ defined in (\ref{UpperLowerBounds}), as opposed to searching among all possible channels.

\label{SystemErrorExponentCC}

\section{Individual User Decoding in Random Multiple Access Communication}
\label{Section IndividualDecoding}
In Section \ref{Section MRAC}, we have assumed that the receiver either decodes messages or reports collisions for {\em all} users in the system. In practical applications, even though many users compete for the wireless channel, it is common that the receiver may not be interested in recovering messages for all of them. In this section, we show that the results obtained in Section \ref{Section MRAC} can help to derive error probability bounds in a random multiple access system where the receiver is only interested in recovering the messages from a user subset. However, to simplify the notations, we will only consider a special case when the communication channel is known at the receiver, and when the receiver is only interested in decoding for a single user. Generalizing the results to decoding for multiple users over a compound channel is straightforward.

Let the discrete-time memoryless channel be characterized by $P_{Y|\mbox{\scriptsize\boldmath$X$}}$, which is known at the receiver. In each time slot, each user chooses a communication rate and encodes its message using the random coding scheme described in Section \ref{Section MRAC}. The rate information is shared neither among the users nor with the receiver. We assume that the receiver is only interested in recovering the message for user $k\in \{1, \cdots, K\}$. We assume that the receiver chooses an operation region $\mathcal{R}$, such that if the transmitted rate vector $\mbox{\boldmath $r$}$ satisfies $\mbox{\boldmath $r$}\in\mathcal{R}$, the receiver intends to decode for user $k$, and if $\mbox{\boldmath $r$}\not\in\mathcal{R}$, the receiver intends to report a collision for user $k$. It is important to note that, first, whether the receiver will be able to decode the message of user $k$, not only depends on the rate of user $k$, but also depends on the rate of other users. Therefore, the operation rate region $\mathcal{R}$ should still be defined as a set of rate vector $\mbox{\boldmath $r$}$, as opposed to the rate of user $k$. Second, even though the receiver only cares about the message of user $k$, the receiver still has the option of decoding the messages for some other users if this helps to improve the communication performance of user $k$. This implies that, based upon the received symbols, the receiver will essentially need to make a decision on which subset of the messages should be decoded.

Due to the above understandings, we first define an elementary decoder, called the ``$(\mathcal{D}, \mathcal{R}_{\mathcal{D}})$-decoder". Given a user subset $\mathcal{D}\subseteq \{1, \cdots, K\}$ and an operation rate region $\mathcal{R}_{\mathcal{D}}$, the ``$(\mathcal{D}, \mathcal{R}_{\mathcal{D}})$-decoder" intends to recover messages for users in $\mathcal{D}$ while regarding signals from users not in $\mathcal{D}$ as interference, if the communication rate vector is within the operation region $\mathcal{R}_{\mathcal{D}}$. We define the following error probabilities for a $(\mathcal{D}, \mathcal{R}_{\mathcal{D}})$-decoder. Conditioned on users in $\mathcal{D}$ transmitting $(\mbox{\boldmath $w$}_{\mathcal{D}}, \mbox{\boldmath $r$}_{\mathcal{D}})$ and users not in $\mathcal{D}$ choosing rate $\mbox{\boldmath $r$}_{\bar{\mathcal{D}}}$, let us denote the estimated messages and rates by $(\hat{\mbox{\boldmath $w$}}_{\mathcal{D}}, \hat{\mbox{\boldmath $r$}}_{\mathcal{D}})$ and $\hat{\mbox{\boldmath $r$}}_{\bar{\mathcal{D}}}$ if the decoder does not report a collision. We define the decoding error probability of the $(\mathcal{D}, \mathcal{R}_{\mathcal{D}})$-decoder for $(\mbox{\boldmath $w$}_{\mathcal{D}}, \mbox{\boldmath $r$}_{\mathcal{D}}, \mbox{\boldmath $r$}_{\bar{\mathcal{D}}})$ with $\mbox{\boldmath $r$} \in \mathcal{R}_{\mathcal{D}}$ as
\begin{eqnarray}
\label{M-DecodeErrorDef}
P_e(\mbox{\boldmath $w$}_{\mathcal{D}}, \mbox{\boldmath $r$}_{\mathcal{D}}, \mbox{\boldmath $r$}_{\bar{\mathcal{D}}})= Pr\{(\hat{\mbox{\boldmath $w$}}_{\mathcal{D}}, \hat{\mbox{\boldmath $r$}}_{\mathcal{D}}) \ne (\mbox{\boldmath $w$}_{\mathcal{D}}, \mbox{\boldmath $r$}_{\mathcal{D}})|(\mbox{\boldmath $w$}_{\mathcal{D}}, \mbox{\boldmath $r$}_{\mathcal{D}}, \mbox{\boldmath $r$}_{\bar{\mathcal{D}}})\}, \quad \forall (\mbox{\boldmath $w$}_{\mathcal{D}}, \mbox{\boldmath $r$}_{\mathcal{D}}, \mbox{\boldmath $r$}_{\bar{\mathcal{D}}}), \mbox{\boldmath $r$}\in \mathcal{R}_{\mathcal{D}}.
\end{eqnarray}
We define the collision miss detection probability for $(\mbox{\boldmath $w$}_{\mathcal{D}}, \mbox{\boldmath $r$}_{\mathcal{D}}, \mbox{\boldmath $r$}_{\bar{\mathcal{D}}})$ with $\mbox{\boldmath $r$} \not\in \mathcal{R}_{\mathcal{D}}$ as
\begin{eqnarray}
\label{M-CollisionErrorDef}
&& \bar{P}_c(\mbox{\boldmath $w$}_{\mathcal{D}}, \mbox{\boldmath $r$}_{\mathcal{D}}, \mbox{\boldmath $r$}_{\bar{\mathcal{D}}}) = 1 - Pr\{ \mbox{``collision"}|(\mbox{\boldmath $w$}_{\mathcal{D}}, \mbox{\boldmath $r$}_{\mathcal{D}}, \mbox{\boldmath $r$}_{\bar{\mathcal{D}}})\}- Pr\{(\hat{\mbox{\boldmath $w$}}_{\mathcal{D}}, \hat{\mbox{\boldmath $r$}}_{\mathcal{D}}) = (\mbox{\boldmath $w$}_{\mathcal{D}}, \mbox{\boldmath $r$}_{\mathcal{D}}) |(\mbox{\boldmath $w$}_{\mathcal{D}}, \mbox{\boldmath $r$}_{\mathcal{D}}, \mbox{\boldmath $r$}_{\bar{\mathcal{D}}})\}, \nonumber\\
&& \qquad \qquad\qquad \qquad\qquad \qquad\qquad \qquad\qquad \qquad\qquad \qquad\qquad \qquad \forall (\mbox{\boldmath $w$}_{\mathcal{D}}, \mbox{\boldmath $r$}_{\mathcal{D}}, \mbox{\boldmath $r$}_{\bar{\mathcal{D}}}), \mbox{\boldmath $r$}\not\in \mathcal{R}_{\mathcal{D}}.
\end{eqnarray}
System error probability of the $(\mathcal{D}, \mathcal{R}_{\mathcal{D}})$-decoder is defined by
\begin{eqnarray}
&& P_{es}(\mathcal{D}, \mathcal{R}_{\mathcal{D}})=\max\left\{\max_{(\mbox{\scriptsize \boldmath $w$}_{\mathcal{D}}, \mbox{\scriptsize \boldmath $r$}_{\mathcal{D}}, \mbox{\scriptsize \boldmath $r$}_{\bar{\mathcal{D}}}), \mbox{\scriptsize \boldmath $r$}\in \mathcal{R}_{\mathcal{D}} } P_e(\mbox{\boldmath $w$}_{\mathcal{D}}, \mbox{\boldmath $r$}_{\mathcal{D}}, \mbox{\boldmath $r$}_{\bar{\mathcal{D}}}), \max_{(\mbox{\scriptsize \boldmath $w$}_{\mathcal{D}}, \mbox{\scriptsize \boldmath $r$}_{\mathcal{D}}, \mbox{\scriptsize \boldmath $r$}_{\bar{\mathcal{D}}}), \mbox{\scriptsize \boldmath $r$} \not\in \mathcal{R}_{\mathcal{D}} } \bar{P}_c(\mbox{\boldmath $w$}_{\mathcal{D}}, \mbox{\boldmath $r$}_{\mathcal{D}}, \mbox{\boldmath $r$}_{\bar{\mathcal{D}}}) \right\}.
\end{eqnarray}

Given a finite codeword length $N$, the following lemma gives an upper bound on the achievable system error probability of a $(\mathcal{D}, \mathcal{R}_{\mathcal{D}})$-decoder.

\begin{lemma}{\label{Lemma1}}
The following system error probability bound is achievable for a $K$-user random multiple access communication system over a discrete-time memoryless channel $P_{Y|\mbox{\scriptsize\boldmath$X$}}$ with an $(\mathcal{D}, \mathcal{R}_{\mathcal{D}})$-decoder,
\begin{eqnarray}
\label{BoundLemma1}
P_{es}(\mathcal{D}, \mathcal{R}_{\mathcal{D}}) \le \max\left\{ \max_{\mbox{\scriptsize \boldmath $r$}\in \mathcal{R}_{\mathcal{D}} }\sum_{\mathcal{S}\subset \mathcal{D}} \left[ \sum_{\tiny \begin{array}{c} \tilde{\mbox{\scriptsize \boldmath $r$}}\in \mathcal{R}_{\mathcal{D}}, \\ \tilde{\mbox{\scriptsize \boldmath $r$}}_{\mathcal{S}}= \mbox{\scriptsize \boldmath $r$}_{\mathcal{S}} \end{array}} \exp\{-N E_{m\mathcal{D}}(\mathcal{S}, \mbox{\boldmath $r$}, \tilde{\mbox{\boldmath $r$}})\}  +  \max_{\tiny \begin{array}{c} \mbox{\scriptsize \boldmath $r$}'\not\in \mathcal{R}_{\mathcal{D}}, \\ \mbox{\scriptsize \boldmath $r$}'_{\mathcal{S}}= \mbox{\scriptsize \boldmath $r$}_{\mathcal{S}} \end{array} } \exp\{-NE_{i\mathcal{D}}(\mathcal{S}, \mbox{\boldmath $r$}, \mbox{\boldmath $r$}') \} \right] \right.,  \nonumber \\
\left.\max_{\tilde{\mbox{\scriptsize \boldmath $r$}}\not\in \mathcal{R}_{\mathcal{D}} }\sum_{\mathcal{S}\subset \mathcal{D}} \sum_{\tiny \begin{array}{c} \mbox{\scriptsize \boldmath $r$}\in \mathcal{R}_{\mathcal{D}},  \\ \mbox{\scriptsize \boldmath $r$}_{\mathcal{S}}=\tilde{\mbox{\scriptsize \boldmath $r$}}_{\mathcal{S}}\end{array} } \max_{\tiny \begin{array}{c} \mbox{\scriptsize \boldmath $r$}'\not\in \mathcal{R}_{\mathcal{D}}, \\ \mbox{\scriptsize \boldmath $r$}'_{\mathcal{S}}= \tilde{\mbox{\scriptsize \boldmath $r$}}_{\mathcal{S}} \end{array} } \exp\{ -N E_{i\mathcal{D}}(\mathcal{S}, \mbox{\boldmath $r$}, \mbox{\boldmath $r$}') \}     \right\}, \nonumber \\
 \label{SRSBound}
\end{eqnarray}
where $E_{m\mathcal{D}}(\mathcal{S}, \mbox{\boldmath $r$}, \tilde{\mbox{\boldmath $r$}})$ and $E_{i\mathcal{D}}(\mathcal{S}, \mbox{\boldmath $r$}, \mbox{\boldmath $r$}')$ are given by,
\begin{eqnarray}
\label{EmEiDDcoder}
&& E_{m\mathcal{D}}(\mathcal{S}, \mbox{\boldmath $r$}, \tilde{\mbox{\boldmath $r$}})= \max_{0<\rho \le 1} -\rho \sum_{k\in \mathcal{D}\setminus \mathcal{S}}\tilde{r}_k \nonumber  + \max_{0<s\le 1} -\log \sum_Y \sum_{\mbox{\scriptsize \boldmath $X$}_{\mathcal{S}}} \prod_{k\in \mathcal{S}} P_{X|r_k}(X_k)                    \nonumber \\
&& \quad \times \left(\sum_{\mbox{\scriptsize \boldmath $X$}_{\mathcal{D}\setminus \mathcal{S}}}\prod_{k \in \mathcal{D}\setminus \mathcal{S}}P_{X|r_k}(X_k)P(Y|\mbox{\boldmath $X$}_{\mathcal{D}}, \mbox{\boldmath $r$}_{\bar{\mathcal{D}}} )^{1-s}\right) \left(\sum_{\mbox{\scriptsize \boldmath $X$}_{\mathcal{D}\setminus \mathcal{S}}}\prod_{k \in \mathcal{D}\setminus \mathcal{S}}P_{X|\tilde{r}_k}(X_k)P(Y|\mbox{\boldmath $X$}_{\mathcal{D}}, \tilde{\mbox{\boldmath $r$}}_{\bar{\mathcal{D}}} )^{\frac{s}{\rho}} \right)^{\rho},                       \nonumber \\
&& E_{i\mathcal{D}}(\mathcal{S}, \mbox{\boldmath $r$}, \mbox{\boldmath $r$}') = \max_{0<\rho \le 1} -\rho \sum_{k\in \mathcal{D}\setminus \mathcal{S}}r_k  + \max_{0<s \le 1-\rho} - \log \sum_Y \sum_{\mbox{\scriptsize \boldmath $X$}_{\mathcal{S}}} \prod_{k\in \mathcal{S}} P_{X|r_k}(X_k)               \nonumber \\
&& \times \left(\sum_{\mbox{\scriptsize \boldmath $X$}_{\mathcal{D}\setminus \mathcal{S}}}\prod_{k \in \mathcal{D}\setminus \mathcal{S}}P_{X|r_k}(X_k)P(Y|\mbox{\boldmath $X$}_{\mathcal{D}}, \mbox{\boldmath $r$}_{\bar{\mathcal{D}}}  )^{\frac{s}{s+\rho}} \right)^{s+\rho}  \left(\sum_{\mbox{\scriptsize \boldmath $X$}_{\mathcal{D}\setminus \mathcal{S}}}\prod_{k \in \mathcal{D}\setminus \mathcal{S}}P_{X|r'_k}(X_k)P(Y|\mbox{\boldmath $X$}_{\mathcal{D}}, \mbox{\boldmath $r$}'_{\bar{\mathcal{D}}})\right)^{1-s},
\end{eqnarray}
with $P(Y|\mbox{\boldmath $X$}_{\mathcal{D}}, \mbox{\boldmath $r$}_{\bar{\mathcal{D}}} )$ in the above equations defined as
\begin{equation}
\label{EquivalentCompound}
P(Y|\mbox{\boldmath $X$}_{\mathcal{D}}, \mbox{\boldmath $r$}_{\bar{\mathcal{D}}} )= \sum_{\mbox{\scriptsize \boldmath $X$}_{\bar{\mathcal{D}}}}\prod_{k \in \bar{\mathcal{D}}}P_{X|r_k}(X_k)P_{Y|\mbox{\scriptsize\boldmath$X$}}(Y|\mbox{\boldmath $X$}).
\end{equation}
$\QED$
\end{lemma}

\begin{proof}
Since the decoder regards signals from users not in $\mathcal{D}$ as interference, given that users not in $\mathcal{D}$ choose rate $\mbox{\boldmath $r$}_{\bar{\mathcal{D}}}$, the multiple access channel experienced by users in $\mathcal{D}$ is characterized by $P(Y|\mbox{\boldmath $X$}_{\mathcal{D}}, \mbox{\boldmath $r$}_{\bar{\mathcal{D}}})$ as specified in (\ref{EquivalentCompound}). The system can therefore be regarded as a random multiple access system with $|\mathcal{D}|$ users communicating over a compound channel characterized by the set $\{P(Y|\mbox{\boldmath $X$}_{\mathcal{D}}, \mbox{\boldmath $r$}_{\bar{\mathcal{D}}} )| \forall \mbox{\boldmath $r$}_{\bar{\mathcal{D}}} \}$. Consequently, Lemma \ref{Lemma1} is implied directly by Theorem \ref{TheoremMC}.
\end{proof}

Next, we will come back to the system where the receiver is only interested in the message of user $k$. We assume that for each user subset $\mathcal{D}\subseteq \{1, \cdots, K\}$ with $k\in \mathcal{D}$, the receiver assigns an operation region $\mathcal{R}_{\mathcal{D}} \subseteq\mathcal{R}$ for the $(\mathcal{D}, \mathcal{R}_{\mathcal{D}})$-decoder. That is, if the transmission rate $\mbox{\boldmath $r$}$ satisfies $\mbox{\boldmath $r$}\in \mathcal{R}_{\mathcal{D}}$, the receiver intends to use the $(\mathcal{D}, \mathcal{R}_{\mathcal{D}})$-decoder to recover the message of user $k$. It is easy to see that we should have,
\begin{equation}
\mathcal{R}=\bigcup_{\mathcal{D}: \mathcal{D}\subseteq\{1, \cdots, K\}, k\in \mathcal{D}} \mathcal{R}_{\mathcal{D}}.
\label{RegionAssignment}
\end{equation}
Assume that the receiver (single-user decoder) carries out all the $(\mathcal{D}, \mathcal{R}_{\mathcal{D}})$-decoding operations. The receiver outputs an estimated message $\hat{w}_k$ for user $k$ if at least one  $(\mathcal{D}, \mathcal{R}_{\mathcal{D}})$-decoder outputs an estimated message, and all estimation outputs of the $(\mathcal{D},\mathcal{R}_{\mathcal{D}})$-decoders for user $k$ are identical. Otherwise, the receiver reports a collision for user $k$.

Let the transmitted rate vector be $\mbox{\boldmath$r$}$, and the transmitted message of user $k$ be $w_k$. We define the decoding error probability $P_{e}(w_k,\mbox{\boldmath$r$})$, the collision miss detection probability $\bar{P}_{c}(w_k,\mbox{\boldmath$r$})$ and the system error probability $P_{es}$ as follows,
\begin{eqnarray}\label{SingleDecodingProbDef}
&&P_{e}(w_k,\mbox{\boldmath$r$}) = Pr \left\{(\hat{w}_k,\hat{r}_k) \neq (w_k,r_k) | (w_k,\mbox{\boldmath$r$}) \right\}, \forall (w_k,\mbox{\boldmath$r$}), \mbox{\boldmath$r$} \in \mathcal{R},\nonumber\\
&&\bar{P}_{c}(w_k,\mbox{\boldmath$r$}) = 1 - Pr \left\{ \mbox{``collision''}|(w_k,\mbox{\boldmath$r$}) \right\} - Pr\left\{ (\hat{w}_k,\hat{r}_k) = (w_k,r_k) |(w_k,\mbox{\boldmath$r$})\right\},\forall  (w_k,\mbox{\boldmath$r$}), \mbox{\boldmath$r$} \notin \mathcal{R} ,\nonumber\\
&&P_{es} = \max\left\{ \max_{(w_k,\mbox{\scriptsize\boldmath$r$}), \mbox{\scriptsize\boldmath$r$} \in \mathcal{R}}P_{e}(w_k,\mbox{\boldmath$r$}), \max_{(w_k,\mbox{\scriptsize\boldmath$r$}), \mbox{\scriptsize\boldmath$r$} \notin \mathcal{R}}\bar{P}_{c}(w_k,\mbox{\boldmath$r$})  \right\}.
\end{eqnarray}
The following theorem gives an upper bound on the achievable system error probability of the single-user decoder.

\begin{theorem}{\label{Theorem2}}
Consider a $K$-user random multiple access system over a discrete-time memoryless channel $P_{Y|\mbox{\scriptsize\boldmath$X$}}$, with the receiver only interested in recovering the message for user $k$. Assume the receiver chooses an operation region $\mathcal{R}$. Let $\sigma$ be an arbitrary partitioning of the operation region $\mathcal{R}$ satisfying
\begin{eqnarray}
&& \mathcal{R}=\bigcup_{\mathcal{D}: \mathcal{D}\subseteq\{1, \cdots, K\}, k\in \mathcal{D}} \mathcal{R}_{\mathcal{D}}, \nonumber \\
&& \mathcal{R}_{\mathcal{D}'}\cap \mathcal{R}_{\mathcal{D}}=\phi, \forall \mathcal{D}, \mathcal{D}'\subseteq\{1, \cdots, K\}, \mathcal{D}' \ne \mathcal{D}, k\in \mathcal{D}, \mathcal{D}'.
\label{RegionPartition}
\end{eqnarray}
System error probability of the single-user decoder is upper-bounded by,
\begin{equation}
P_{es} \le \min_{\sigma} \sum_{\mathcal{D}: \mathcal{D}\subseteq\{1, \cdots, K\}, k\in \mathcal{D}} P_{es}(\mathcal{D}, \mathcal{R}_{\mathcal{D}}),
\label{RegionPartitioning}
\end{equation}
where $P_{es}(\mathcal{D}, \mathcal{R}_{\mathcal{D}})$ is the system error probability bound of the $(\mathcal{D}, \mathcal{R}_{\mathcal{D}})$-decoder, and can be further bounded by (\ref{SRSBound}). $\QED$
\end{theorem}

\begin{proof}
Because a $(\mathcal{D}, \mathcal{R}_{\mathcal{D}})$-decoder can always choose to report a collision even if it can decode the messages, its system error probability can be improved by shrinking the operation region $\mathcal{R}_{\mathcal{D}}$. This implies that the receiver of the random access system should partition its operation region $\mathcal{R}$ into $\mathcal{R}_{\mathcal{D}}$ regions that do not overlap with each other. In other words, replacing (\ref{RegionAssignment}) by (\ref{RegionPartition}) will improve the system error performance. The rest of the proof is implied by Lemma \ref{Lemma1}.
\end{proof}

Note that the system error probability bound provided in Theorem \ref{Theorem2} is implicit since the optimal partitioning scheme $\sigma$ that maximize the right hand side of (\ref{RegionPartitioning}) is not specified. To find the optimal partitioning, one essentially needs to compute every single term on the right hand side of (\ref{RegionPartitioning}) and (\ref{BoundLemma1}) for all rate options and all user subsets. Because both $E_{m\mathcal{D}}(\mathcal{S}, \mbox{\boldmath $r$}, \tilde{\mbox{\boldmath $r$}})$ and $E_{i\mathcal{D}}(\mathcal{S}, \mbox{\boldmath $r$}, \mbox{\boldmath $r$}')$ defined in (\ref{EmEiDDcoder}) involve the combinations of two user subsets and two rate vectors, the computational complexity of finding the optimal partitioning scheme is $O\left((2M)^{2K}\right)$.

\section{Conclusions}
\label{SectionConslusions}
We investigated the error performance of the random multiple access system over a compound discrete-time memoryless channel. An achievable bound on the system error probability was derived under the non-asymptotic assumption of a finite codeword length. We showed that the results can be extended to the random multiple access system where the receiver is only interested in decoding messages for a user subset.

\appendix
\subsection{Proof of Theorem \ref{TheoremMC}}
\label{AppendixTheoremMC}
We assume that the following decoding algorithm is used at the receiver. Given the received channel output symbols $\mbox{\boldmath$y$}$, the receiver outputs a message and rate vector pair $(\mbox{\boldmath$w$},\mbox{\boldmath$r$})$ together with a channel realization $P_{Y|\mbox{\scriptsize\boldmath$X$}}$ such that $(\mbox{\boldmath$r$},P_{Y|\mbox{\scriptsize\boldmath$X$}}) \in \mathcal{R}$ if the following condition is satisfied, for all user subsets $\mathcal{S} \subset \{1,\cdots,K\}$,
\begin{eqnarray}
\label{CompoundCriterionAppendix}
&& -\frac{1}{N}\log Pr\{\mbox{\boldmath $y$}|\mbox{\boldmath $x$}_{(\mbox{\scriptsize\boldmath$w$},\mbox{\scriptsize\boldmath$r$})},P_{Y|\mbox{\scriptsize\boldmath$X$}}\} < -\frac{1}{N}\log Pr\{\mbox{\boldmath $y$}|\mbox{\boldmath $x$}_{(\tilde{\mbox{\scriptsize\boldmath$w$}}, \tilde{\mbox{\scriptsize\boldmath$r$}})},\tilde{P}_{Y|\mbox{\scriptsize\boldmath$X$}}\}, \nonumber \\
&&\quad \mbox{ for all } (\tilde{\mbox{\boldmath$w$}},\tilde{\mbox{\boldmath$r$}},\tilde{P}_{Y|\mbox{\scriptsize\boldmath$X$}}), (\tilde{\mbox{\boldmath$w$}}_{\mathcal{S}},\tilde{\mbox{\boldmath$r$}}_{\mathcal{S}}) = (\mbox{\boldmath$w$}_{\mathcal{S}}, \mbox{\boldmath$r$}_{\mathcal{S}}),(\tilde{w}_{k},\tilde{r}_k)\neq (w_k,r_k), \forall k \notin \mathcal{S},\nonumber\\
&&\quad \mbox{ and }(\tilde{\mbox{\boldmath$w$}}, \tilde{\mbox{\boldmath$r$}},\tilde{P}_{Y|\mbox{\scriptsize\boldmath$X$}}), (\mbox{\boldmath$w$}, \mbox{\boldmath$r$}, P_{Y|\mbox{\scriptsize\boldmath$X$}}) \in \mathcal{R}_{(\mathcal{S},\mbox{\scriptsize \boldmath $y$})}, \mbox{ with} \nonumber\\
&& \mathcal{R}_{(\mathcal{S},\mbox{\scriptsize \boldmath $y$})}=\left\{\left.(\tilde{\mbox{\boldmath$w$}}, \tilde{\mbox{\boldmath$r$}},\tilde{P}_{Y|\mbox{\scriptsize\boldmath$X$}})\right| (\tilde{\mbox{\boldmath $r$}},\tilde{P}_{Y|\mbox{\scriptsize\boldmath$X$}})\in \mathcal{R},  -\frac{1}{N}\log Pr\{\mbox{\boldmath $y$}|\mbox{\boldmath $x$}_{(\tilde{\mbox{\scriptsize\boldmath$w$}}, \tilde{\mbox{\scriptsize\boldmath$r$}})}, \tilde{P}_{Y|\mbox{\scriptsize\boldmath$X$}}\} < \tau_{(\tilde{\mbox{\scriptsize\boldmath$r$}}, \tilde{P}_{Y|\mbox{\tiny\boldmath$X$}},\mathcal{S})}(\mbox{\boldmath$y$}) \right\},
\end{eqnarray}
where $\tau_{(\tilde{\mbox{\scriptsize\boldmath$r$}}, \tilde{P}_{Y|\mbox{\tiny\boldmath$X$}},\mathcal{S})}(\cdot)$ is a per-determined typicality threshold function of the channel output symbols $\mbox{\boldmath$y$}$, associated with the rate and channel realization pair $(\tilde{\mbox{\boldmath$r$}}, \tilde{P}_{Y|\mbox{\scriptsize\boldmath$X$}})$ and the user subset $\mathcal{S}$. If there is no codeword satisfying (\ref{CompoundCriterionAppendix}), the receiver reports a collision. In other words, for a given $\mathcal{S}$, the receiver searches for the subset of codewords with likelihood values larger than the corresponding typicality threshold. If the subset is not empty, the receiver outputs the codeword with the maximum likelihood value as the estimate for this given $\mathcal{S}$. If the estimates for all $\mathcal{S} \subset \{1,\cdots, K\}$ agree with each other, the receiver regards this estimate as the decoding decision and outputs the corresponding decoded message and rate pair. Otherwise, the receiver reports a collision. Note that in (\ref{CompoundCriterionAppendix}), for given $\mathcal{S}$ and ($\mbox{\boldmath$w$},\mbox{\boldmath$r$}$), we only compare the likelihood value of codeword vector $\mbox{\boldmath$x$}_{(\mbox{\scriptsize\boldmath$w$},\mbox{\scriptsize\boldmath$r$})}$ with those of the codeword vectors satisfying $(\tilde{\mbox{\boldmath$w$}}_{\mathcal{S}},\tilde{\mbox{\boldmath$r$}}_{\mathcal{S}}) = (\mbox{\boldmath$w$}_{\mathcal{S}}, \mbox{\boldmath$r$}_{\mathcal{S}}),(\tilde{w}_{k},\tilde{r}_k)\neq (w_k,r_k), \forall k \notin \mathcal{S}$. We will first analyze the error performance for each user subset $\mathcal{S}$ and then derive the overall error performance by taking the union over all $\mathcal{S}$.

Given a user subset $\mathcal{S} \subset \{1,\cdots,K\}$, we define the following probability terms.

First, assume $(\mbox{\boldmath$w$},\mbox{\boldmath$r$})$ is transmitted over channel $P_{Y|\mbox{\scriptsize\boldmath$X$}}$, with $(\mbox{\boldmath$r$},P_{Y|\mbox{\scriptsize\boldmath$X$}}) \in \mathcal{R}$. Let $P_{t[\mbox{\scriptsize\boldmath$r$},P_{Y|\mbox{\tiny\boldmath$X$}},\mathcal{S}]}$ be the probability that the likelihood value of the transmitted codeword vector over the channel $P_{Y|\mbox{\scriptsize\boldmath$X$}}$ is no larger than the corresponding typicality threshold,
\begin{eqnarray}
\label{ProofPtCC}
P_{t[\mbox{\scriptsize\boldmath$r$},P_{Y|\mbox{\tiny\boldmath$X$}},\mathcal{S}]} = Pr \left\{ P(\mbox{\boldmath $y$}|\mbox{\boldmath $x$}_{(\mbox{\scriptsize\boldmath$w$},\mbox{\scriptsize\boldmath$r$})},P_{Y|\mbox{\scriptsize\boldmath$X$}}) \le e^{-N\tau_{(\mbox{\tiny\boldmath$r$}, P_{Y|\mbox{\tiny\boldmath$X$}},\mathcal{S})}(\mbox{\scriptsize\boldmath$y$})}\right\}.
\end{eqnarray}
Define $P_{m[(\mbox{\scriptsize\boldmath$r$},P_{Y|\mbox{\tiny\boldmath$X$}}),(\tilde{\mbox{\scriptsize\boldmath$r$}},\tilde{P}_{Y|\mbox{\tiny\boldmath$X$}}),\mathcal{S}]}$ as the probability that the likelihood value of the transmitted codeword vector over the channel realization $P_{Y|\mbox{\scriptsize\boldmath$X$}}$ is no larger than that of another codeword $(\tilde{\mbox{\boldmath$w$}},\tilde{\mbox{\boldmath$r$}})$ with $(\tilde{\mbox{\boldmath$w$}}_{\mathcal{S}},\tilde{\mbox{\boldmath$r$}}_{\mathcal{S}}) = (\mbox{\boldmath$w$}_{\mathcal{S}}, \mbox{\boldmath$r$}_{\mathcal{S}}),(\tilde{w}_{k},\tilde{r}_k)\neq (w_k,r_k), \forall k \notin \mathcal{S}$, over channel $\tilde{P}_{Y|\mbox{\scriptsize\boldmath$X$}}$ with $(\tilde{\mbox{\boldmath$r$}},\tilde{P}_{Y|\mbox{\scriptsize\boldmath$X$}}) \in \mathcal{R}$,
\begin{eqnarray}
\label{ProofPmMC}
&& P_{m[(\mbox{\scriptsize\boldmath$r$},P_{Y|\mbox{\tiny\boldmath$X$}}),(\tilde{\mbox{\scriptsize\boldmath$r$}},\tilde{P}_{Y|\mbox{\tiny\boldmath$X$}}),\mathcal{S}]} = Pr \left\{ P(\mbox{\boldmath $y$}|\mbox{\boldmath $x$}_{(\mbox{\scriptsize\boldmath$w$},\mbox{\scriptsize\boldmath$r$})},P_{Y|\mbox{\scriptsize\boldmath$X$}}) \le P(\mbox{\boldmath $y$}|\mbox{\boldmath $x$}_{(\tilde{\mbox{\scriptsize\boldmath$w$}}, \tilde{\mbox{\scriptsize\boldmath$r$}})},\tilde{P}_{Y|\mbox{\scriptsize\boldmath$X$}})\right\}\nonumber\\
&& \quad (\tilde{\mbox{\boldmath$w$}},\tilde{\mbox{\boldmath$r$}},\tilde{P}_{Y|\mbox{\scriptsize\boldmath$X$}}),(\tilde{\mbox{\boldmath$r$}},\tilde{P}_{Y|\mbox{\scriptsize\boldmath$X$}}) \in \mathcal{R}, (\tilde{\mbox{\boldmath$w$}}_{\mathcal{S}},\tilde{\mbox{\boldmath$r$}}_{\mathcal{S}}) = (\mbox{\boldmath$w$}_{\mathcal{S}}, \mbox{\boldmath$r$}_{\mathcal{S}}),(\tilde{w}_{k},\tilde{r}_k)\neq (w_k,r_k), \forall k \notin \mathcal{S}.
\end{eqnarray}

Second, assume that $(\tilde{\mbox{\boldmath$w$}},\tilde{\mbox{\boldmath$r$}})$ is transmitted over channel $\tilde{P}_{Y|\mbox{\scriptsize\boldmath$X$}}$, with $(\tilde{\mbox{\boldmath$r$}},\tilde{P}_{Y|\mbox{\scriptsize\boldmath$X$}}) \notin \mathcal{R}$. Define $P_{i[(\tilde{\mbox{\scriptsize\boldmath$r$}},\tilde{P}_{Y|\mbox{\tiny\boldmath$X$}}),(\mbox{\scriptsize\boldmath$r$},P_{Y|\mbox{\tiny\boldmath$X$}}),\mathcal{S}]}$ as the probability that the decoder finds a codeword $(\mbox{\boldmath$w$},\mbox{\boldmath$r$})$ with $ (\mbox{\boldmath$w$}_{\mathcal{S}}, \mbox{\boldmath$r$}_{\mathcal{S}})= (\tilde{\mbox{\boldmath$w$}}_{\mathcal{S}},\tilde{\mbox{\boldmath$r$}}_{\mathcal{S}}), (w_k,r_k) \neq (\tilde{w}_{k},\tilde{r}_k), \forall k \notin \mathcal{S} $, over channel $P_{Y|\mbox{\scriptsize\boldmath$X$}}$ with $(\mbox{\boldmath$r$},P_{Y|\mbox{\scriptsize\boldmath$X$}}) \in \mathcal{R}$, such that its likelihood value is larger than the corresponding typicality threshold,
\begin{eqnarray}
\label{ProofPiCC}
&& P_{i[(\tilde{\mbox{\scriptsize\boldmath$r$}},\tilde{P}_{Y|\mbox{\tiny\boldmath$X$}}),(\mbox{\scriptsize\boldmath$r$},P_{Y|\mbox{\tiny\boldmath$X$}}),\mathcal{S}]} = Pr \left\{ P(\mbox{\boldmath $y$}|\mbox{\boldmath $x$}_{(\mbox{\scriptsize\boldmath$w$},\mbox{\scriptsize\boldmath$r$})},P_{Y|\mbox{\scriptsize\boldmath$X$}}) > e^{-N\tau_{(\mbox{\tiny\boldmath$r$}, P_{Y|\mbox{\tiny\boldmath$X$}},\mathcal{S})}(\mbox{\scriptsize\boldmath$y$})}\right\},\nonumber\\
&& \quad (\mbox{\boldmath$w$},\mbox{\boldmath$r$},P_{Y|\mbox{\scriptsize\boldmath$X$}}),(\mbox{\boldmath$r$},P_{Y|\mbox{\scriptsize\boldmath$X$}}) \in \mathcal{R}, (\mbox{\boldmath$w$}_{\mathcal{S}}, \mbox{\boldmath$r$}_{\mathcal{S}})= (\tilde{\mbox{\boldmath$w$}}_{\mathcal{S}},\tilde{\mbox{\boldmath$r$}}_{\mathcal{S}}), (w_k,r_k) \neq (\tilde{w}_{k},\tilde{r}_k), \forall k \notin \mathcal{S}.
\end{eqnarray}

With the above probability definitions, by applying the union bound over all $\mathcal{S}$, we can upper-bound the system error probability by
\begin{eqnarray}
\label{ProofPes}
P_{es} \le \max \left\{ \begin{array}{l} \max_{(\mbox{\scriptsize\boldmath$r$},P_{Y|\mbox{\tiny\boldmath$X$}}) \in \mathcal{R}} \sum_{\mathcal{S} \subset \{1,\cdots,K\}} \left[ P_{t[\mbox{\scriptsize\boldmath$r$},P_{Y|\mbox{\tiny\boldmath$X$}},\mathcal{S}]} + \sum_{(\tilde{\mbox{\scriptsize\boldmath$r$}},\tilde{P}_{Y|\mbox{\tiny\boldmath$X$}}) \in \mathcal{R},\tilde{\mbox{\scriptsize\boldmath$r$}}_{\mathcal{S}} = \mbox{\scriptsize\boldmath$r$}_{\mathcal{S}}} P_{m[(\mbox{\scriptsize\boldmath$r$},P_{Y|\mbox{\tiny\boldmath$X$}}),(\tilde{\mbox{\scriptsize\boldmath$r$}},\tilde{P}_{Y|\mbox{\tiny\boldmath$X$}}),\mathcal{S}]} \right] , \\
\max_{(\tilde{\mbox{\scriptsize\boldmath$r$}},\tilde{P}_{Y|\mbox{\tiny\boldmath$X$}}) \notin \mathcal{R}} \sum_{\mathcal{S} \subset \{1,\cdots,K\}}
\sum_{(\mbox{\scriptsize\boldmath$r$},P_{Y|\mbox{\tiny\boldmath$X$}}) \in \mathcal{R},\mbox{\scriptsize\boldmath$r$}_{\mathcal{S}} = \tilde{\mbox{\scriptsize\boldmath$r$}}_{\mathcal{S}}} P_{i[(\tilde{\mbox{\scriptsize\boldmath$r$}},\tilde{P}_{Y|\mbox{\tiny\boldmath$X$}}),(\mbox{\scriptsize\boldmath$r$},P_{Y|\mbox{\tiny\boldmath$X$}}),\mathcal{S}]} \end{array}\right\}.
\end{eqnarray}
Next, we will derive individual upper-bounds for each of the probability terms on the right hand side of (\ref{ProofPes}).

{\bf Step I: } Upper-bounding $P_{m[(\mbox{\scriptsize\boldmath$r$},P_{Y|\mbox{\tiny\boldmath$X$}}),(\tilde{\mbox{\scriptsize\boldmath$r$}},\tilde{P}_{Y|\mbox{\tiny\boldmath$X$}}),\mathcal{S}]} $

Denote $E_{\theta}$ as the expectation operator over random variable $\theta$ which is defined in Section \ref{Section MRAC}. Consequently, given $(\mbox{\boldmath$r$},P_{Y|\mbox{\scriptsize\boldmath$X$}}),(\tilde{\mbox{\boldmath$r$}},\tilde{P}_{Y|\mbox{\scriptsize\boldmath$X$}}) \in \mathcal{R}$, $P_{m[(\mbox{\scriptsize\boldmath$r$},P_{Y|\mbox{\tiny\boldmath$X$}}),(\tilde{\mbox{\scriptsize\boldmath$r$}},\tilde{P}_{Y|\mbox{\tiny\boldmath$X$}}),\mathcal{S}]} $ defined in (\ref{ProofPmMC}) can be rewritten as
\begin{eqnarray}
\label{PmBound1}
P_{m[(\mbox{\scriptsize\boldmath$r$},P_{Y|\mbox{\tiny\boldmath$X$}}),(\tilde{\mbox{\scriptsize\boldmath$r$}},\tilde{P}_{Y|\mbox{\tiny\boldmath$X$}}),\mathcal{S}]} = E_{\theta} \left[\sum_{\mbox{\scriptsize\boldmath$y$}} P(\mbox{\boldmath $y$}|\mbox{\boldmath $x$}_{(\mbox{\scriptsize\boldmath$w$},\mbox{\scriptsize\boldmath$r$})},P_{Y|\mbox{\scriptsize\boldmath$X$}}) \phi_{m[(\mbox{\scriptsize\boldmath$r$},P_{Y|\mbox{\tiny\boldmath$X$}}),(\tilde{\mbox{\scriptsize\boldmath$r$}},\tilde{P}_{Y|\mbox{\tiny\boldmath$X$}}),\mathcal{S}]} (\mbox{\boldmath$y$}) \right],
\end{eqnarray}
where $\phi_{m[(\mbox{\scriptsize\boldmath$r$},P_{Y|\mbox{\tiny\boldmath$X$}}),(\tilde{\mbox{\scriptsize\boldmath$r$}},\tilde{P}_{Y|\mbox{\tiny\boldmath$X$}}),\mathcal{S}]} (\mbox{\boldmath$y$}) = 1$ if $P(\mbox{\boldmath $y$}|\mbox{\boldmath $x$}_{(\mbox{\scriptsize\boldmath$w$},\mbox{\scriptsize\boldmath$r$})},P_{Y|\mbox{\scriptsize\boldmath$X$}}) \le P(\mbox{\boldmath $y$}|\mbox{\boldmath $x$}_{(\tilde{\mbox{\scriptsize\boldmath$w$}}, \tilde{\mbox{\scriptsize\boldmath$r$}})},\tilde{P}_{Y|\mbox{\scriptsize\boldmath$X$}})$ for some triplet $(\tilde{\mbox{\boldmath$w$}},\tilde{\mbox{\boldmath$r$}},\tilde{P}_{Y|\mbox{\scriptsize\boldmath$X$}})$ with $(\tilde{\mbox{\boldmath$r$}},\tilde{P}_{Y|\mbox{\scriptsize\boldmath$X$}}) \in \mathcal{R}, (\tilde{\mbox{\boldmath$w$}}_{\mathcal{S}},\tilde{\mbox{\boldmath$r$}}_{\mathcal{S}}) = (\mbox{\boldmath$w$}_{\mathcal{S}}, \mbox{\boldmath$r$}_{\mathcal{S}}),(\tilde{w}_{k},\tilde{r}_k)\neq (w_k,r_k), \forall k \notin \mathcal{S}$. Otherwise, $\phi_{m[(\mbox{\scriptsize\boldmath$r$},P_{Y|\mbox{\tiny\boldmath$X$}}),(\tilde{\mbox{\scriptsize\boldmath$r$}},\tilde{P}_{Y|\mbox{\tiny\boldmath$X$}}),\mathcal{S}]} (\mbox{\boldmath$y$}) = 0$. We can upper-bound $\phi_{m[(\mbox{\scriptsize\boldmath$r$},P_{Y|\mbox{\tiny\boldmath$X$}}),(\tilde{\mbox{\scriptsize\boldmath$r$}},\tilde{P}_{Y|\mbox{\tiny\boldmath$X$}}),\mathcal{S}]} (\mbox{\boldmath$y$}) $ for any constants $\rho>0$ and $s>0$ as follows,
\begin{eqnarray}
\label{PhiBoundPmCC}
\phi_{m[(\mbox{\scriptsize\boldmath$r$},P_{Y|\mbox{\tiny\boldmath$X$}}),(\tilde{\mbox{\scriptsize\boldmath$r$}},\tilde{P}_{Y|\mbox{\tiny\boldmath$X$}}),\mathcal{S}]} (\mbox{\boldmath$y$}) \le \left[ \frac{\sum_{\tilde{\mbox{\scriptsize\boldmath$w$}}, (\tilde{\mbox{\scriptsize\boldmath$w$}}_{\mathcal{S}},\tilde{\mbox{\scriptsize\boldmath$r$}}_{\mathcal{S}}) = (\mbox{\scriptsize\boldmath$w$}_{\mathcal{S}}, \mbox{\scriptsize\boldmath$r$}_{\mathcal{S}}),(\tilde{w}_{k},\tilde{r}_k)\neq (w_k,r_k), \forall k \notin \mathcal{S}} P(\mbox{\boldmath $y$}|\mbox{\boldmath $x$}_{(\tilde{\mbox{\scriptsize\boldmath$w$}},\tilde{\mbox{\scriptsize\boldmath$r$}})},\tilde{P}_{Y|\mbox{\scriptsize\boldmath$X$}})^{\frac{s}{\rho}} } {P(\mbox{\boldmath $y$}|\mbox{\boldmath $x$}_{(\mbox{\scriptsize\boldmath$w$},\mbox{\scriptsize\boldmath$r$})},P_{Y|\mbox{\scriptsize\boldmath$X$}})^{\frac{s}{\rho}} }  \right]^{\rho}.
\end{eqnarray}
Substituting (\ref{PhiBoundPmCC}) back into (\ref{PmBound1}) gives,
\begin{eqnarray}
\label{PmBound2}
P_{m[(\mbox{\scriptsize\boldmath$r$},P_{Y|\mbox{\tiny\boldmath$X$}}),(\tilde{\mbox{\scriptsize\boldmath$r$}},\tilde{P}_{Y|\mbox{\tiny\boldmath$X$}}),\mathcal{S}]} &\le& E_{\theta} \left[ \sum_{\mbox{\scriptsize\boldmath$y$}} P(\mbox{\boldmath $y$}|\mbox{\boldmath $x$}_{(\mbox{\scriptsize\boldmath$w$},\mbox{\scriptsize\boldmath$r$})},P_{Y|\mbox{\scriptsize\boldmath$X$}}) \right.\nonumber\\
&& \left. \times \left[ \frac{\sum_{\tilde{\mbox{\scriptsize\boldmath$w$}}, (\tilde{\mbox{\scriptsize\boldmath$w$}}_{\mathcal{S}},\tilde{\mbox{\scriptsize\boldmath$r$}}_{\mathcal{S}}) = (\mbox{\scriptsize\boldmath$w$}_{\mathcal{S}}, \mbox{\scriptsize\boldmath$r$}_{\mathcal{S}}),(\tilde{w}_{k},\tilde{r}_k)\neq (w_k,r_k), \forall k \notin \mathcal{S}} P(\mbox{\boldmath $y$}|\mbox{\boldmath $x$}_{(\tilde{\mbox{\scriptsize\boldmath$w$}},\tilde{\mbox{\scriptsize\boldmath$r$}})},\tilde{P}_{Y|\mbox{\scriptsize\boldmath$X$}})^{\frac{s}{\rho}} } {P(\mbox{\boldmath $y$}|\mbox{\boldmath $x$}_{(\mbox{\scriptsize\boldmath$w$},\mbox{\scriptsize\boldmath$r$})},P_{Y|\mbox{\scriptsize\boldmath$X$}})^{\frac{s}{\rho}} }  \right]^{\rho}
\right] \nonumber\\
& = & \sum_{\mbox{\scriptsize\boldmath$y$}} E_{\theta_{\mathcal{S}}} \left[  E_{\theta_{\bar{\mathcal{S}}}} \left[  P(\mbox{\boldmath $y$}|\mbox{\boldmath $x$}_{(\mbox{\scriptsize\boldmath$w$},\mbox{\scriptsize\boldmath$r$})},P_{Y|\mbox{\scriptsize\boldmath$X$}}) ^{1-s} \right] \right.\nonumber\\
&& \left.\times E_{\theta_{\bar{\mathcal{S}}}} \left[ \left[ \sum_{\tilde{\mbox{\scriptsize\boldmath$w$}}, (\tilde{\mbox{\scriptsize\boldmath$w$}}_{\mathcal{S}},\tilde{\mbox{\scriptsize\boldmath$r$}}_{\mathcal{S}}) = (\mbox{\scriptsize\boldmath$w$}_{\mathcal{S}}, \mbox{\scriptsize\boldmath$r$}_{\mathcal{S}}),(\tilde{w}_{k},\tilde{r}_k)\neq (w_k,r_k), \forall k \notin \mathcal{S}} P(\mbox{\boldmath $y$}|\mbox{\boldmath $x$}_{(\tilde{\mbox{\scriptsize\boldmath$w$}},\tilde{\mbox{\scriptsize\boldmath$r$}})},\tilde{P}_{Y|\mbox{\scriptsize\boldmath$X$}})^{\frac{s}{\rho}}  \right]^{\rho} \right]\right].
\end{eqnarray}
The second step in (\ref{PmBound2}) is due to the independence between the codewords corresponding to $(\mbox{\boldmath$w$}_{\bar{\mathcal{S}}},\mbox{\boldmath$r$}_{\bar{\mathcal{S}}})$ and $(\tilde{\mbox{\boldmath$w$}}_{\bar{\mathcal{S}}},\tilde{\mbox{\boldmath$r$}}_{\bar{\mathcal{S}}})$.

With the assumption of $0<\rho\le 1$, we can further bound $P_{m[(\mbox{\scriptsize\boldmath$r$},P_{Y|\mbox{\tiny\boldmath$X$}}),(\tilde{\mbox{\scriptsize\boldmath$r$}},\tilde{P}_{Y|\mbox{\tiny\boldmath$X$}}),\mathcal{S}]} $ by
\begin{eqnarray}
\label{PmBound3}
&&P_{m[(\mbox{\scriptsize\boldmath$r$},P_{Y|\mbox{\tiny\boldmath$X$}}),(\tilde{\mbox{\scriptsize\boldmath$r$}},\tilde{P}_{Y|\mbox{\tiny\boldmath$X$}}),\mathcal{S}]} \le \sum_{\mbox{\scriptsize\boldmath$y$}} E_{\theta_{\mathcal{S}}} \left[  E_{\theta_{\bar{\mathcal{S}}}} \left[  P(\mbox{\boldmath $y$}|\mbox{\boldmath $x$}_{(\mbox{\scriptsize\boldmath$w$},\mbox{\scriptsize\boldmath$r$})},P_{Y|\mbox{\scriptsize\boldmath$X$}}) ^{1-s} \right] \right.\nonumber\\
&&\qquad\qquad\qquad\qquad \left. \times E_{\theta_{\bar{\mathcal{S}}}} \left[ \left [\sum_{\tilde{\mbox{\scriptsize\boldmath$w$}}, (\tilde{\mbox{\scriptsize\boldmath$w$}}_{\mathcal{S}},\tilde{\mbox{\scriptsize\boldmath$r$}}_{\mathcal{S}}) = (\mbox{\scriptsize\boldmath$w$}_{\mathcal{S}}, \mbox{\scriptsize\boldmath$r$}_{\mathcal{S}})} P(\mbox{\boldmath $y$}|\mbox{\boldmath $x$}_{(\tilde{\mbox{\scriptsize\boldmath$w$}},\tilde{\mbox{\scriptsize\boldmath$r$}})},\tilde{P}_{Y|\mbox{\scriptsize\boldmath$X$}})^{\frac{s}{\rho}}  \right]^{\rho} \right] \right] \nonumber\\
&& \le e^{N\rho \sum_{k\notin\mathcal{S}} \tilde{r}_k} \sum_{\mbox{\scriptsize\boldmath$y$}} E_{\theta_{\mathcal{S}}} \left[  E_{\theta_{\bar{\mathcal{S}}}} \left[ P(\mbox{\boldmath $y$}|\mbox{\boldmath $x$}_{(\mbox{\scriptsize\boldmath$w$},\mbox{\scriptsize\boldmath$r$})},P_{Y|\mbox{\scriptsize\boldmath$X$}}) ^{1-s} \right] E_{\theta_{\bar{\mathcal{S}}}} \left[ \left[  P(\mbox{\boldmath $y$}|\mbox{\boldmath $x$}_{(\tilde{\mbox{\scriptsize\boldmath$w$}},\tilde{\mbox{\scriptsize\boldmath$r$}})},\tilde{P}_{Y|\mbox{\scriptsize\boldmath$X$}})^{\frac{s}{\rho}}  \right]^{\rho} \right] \right].
\end{eqnarray}

It is easy to verify that the bound in (\ref{PmBound3}) holds for all $0<\rho\le 1$ and $s>0$, and becomes trivial for $s>1$. Consequently, (\ref{PmBound3}) gives the following upper bound,
\begin{eqnarray}
\label{PmBound4}
P_{m[(\mbox{\scriptsize\boldmath$r$},P_{Y|\mbox{\tiny\boldmath$X$}}),(\tilde{\mbox{\scriptsize\boldmath$r$}},\tilde{P}_{Y|\mbox{\tiny\boldmath$X$}}),\mathcal{S}]} \le \exp \left\{ -NE_m(\mathcal{S}, \mbox{\boldmath$r$}, \tilde{\mbox{\boldmath$r$}}, P_{Y| \mbox{\scriptsize\boldmath$X$}},\tilde{P}_{Y| \mbox{\scriptsize\boldmath$X$}}) \right\},
\end{eqnarray}
where $E_m(\mathcal{S}, \mbox{\boldmath$r$}, \tilde{\mbox{\boldmath$r$}}, P_{Y| \mbox{\scriptsize\boldmath$X$}},\tilde{P}_{Y| \mbox{\scriptsize\boldmath$X$}})$ is specified in (\ref{EmEiMultiMC}).

{\bf Step II: } Upper-bounding $P_{t[\mbox{\scriptsize\boldmath$r$},P_{Y|\mbox{\tiny\boldmath$X$}},\mathcal{S}]}$

Given that $(\mbox{\boldmath$r$},P_{Y|\mbox{\scriptsize\boldmath$X$}}) \in \mathcal{R}$, we can rewrite $P_{t[\mbox{\scriptsize\boldmath$r$},P_{Y|\mbox{\tiny\boldmath$X$}},\mathcal{S}]} $, defined in (\ref{ProofPtCC}), as follows,
\begin{eqnarray}
\label{PtBound1}
P_{t[\mbox{\scriptsize\boldmath$r$},P_{Y|\mbox{\tiny\boldmath$X$}},\mathcal{S}]} = E_{\theta} \left[\sum_{\mbox{\scriptsize\boldmath$y$}} P(\mbox{\boldmath $y$}|\mbox{\boldmath $x$}_{(\mbox{\scriptsize\boldmath$w$},\mbox{\scriptsize\boldmath$r$})},P_{Y|\mbox{\scriptsize\boldmath$X$}}) \phi_{t[\mbox{\scriptsize\boldmath$r$},P_{Y|\mbox{\tiny\boldmath$X$}},\mathcal{S}]} (\mbox{\boldmath$y$}) \right],
\end{eqnarray}
where $\phi_{t[\mbox{\scriptsize\boldmath$r$},P_{Y|\mbox{\tiny\boldmath$X$}},\mathcal{S}]} (\mbox{\boldmath$y$}) = 1$ if $P(\mbox{\boldmath $y$}|\mbox{\boldmath $x$}_{(\mbox{\scriptsize\boldmath$w$},\mbox{\scriptsize\boldmath$r$})},P_{Y|\mbox{\scriptsize\boldmath$X$}}) \le e^{-N\tau_{(\mbox{\tiny\boldmath$r$}, P_{Y|\mbox{\tiny\boldmath$X$}},\mathcal{S})}(\mbox{\scriptsize\boldmath$y$})}$, otherwise $\phi_{t[\mbox{\scriptsize\boldmath$r$},P_{Y|\mbox{\tiny\boldmath$X$}},\mathcal{S}]} (\mbox{\boldmath$y$}) = 0$. Note that the value of $\tau_{(\mbox{\tiny\boldmath$r$}, P_{Y|\mbox{\tiny\boldmath$X$}},\mathcal{S})}(\mbox{\scriptsize\boldmath$y$})$ will be determined in Step IV. Similarly, we can bound $\phi_{t[\mbox{\scriptsize\boldmath$r$},P_{Y|\mbox{\tiny\boldmath$X$}},\mathcal{S}]} (\mbox{\boldmath$y$})$, for any $s_1>0$, as follows,
\begin{eqnarray}
\phi_{t[\mbox{\scriptsize\boldmath$r$},P_{Y|\mbox{\tiny\boldmath$X$}},\mathcal{S}]} (\mbox{\boldmath$y$}) \le
\frac{e^{-Ns_1\tau_{(\mbox{\tiny\boldmath$r$}, P_{Y|\mbox{\tiny\boldmath$X$}},\mathcal{S})}(\mbox{\scriptsize\boldmath$y$})}}
{P(\mbox{\boldmath $y$}|\mbox{\boldmath $x$}_{(\mbox{\scriptsize\boldmath$w$},\mbox{\scriptsize\boldmath$r$})},P_{Y|\mbox{\scriptsize\boldmath$X$}})^{s_1}}.
\end{eqnarray}
This yields,
\begin{eqnarray}
\label{MCBound4P_t}
P_{t[\mbox{\scriptsize\boldmath$r$},P_{Y|\mbox{\tiny\boldmath$X$}},\mathcal{S}]} &\le& E_{\theta} \left[\sum_{\mbox{\scriptsize\boldmath$y$}} P(\mbox{\boldmath $y$}|\mbox{\boldmath $x$}_{(\mbox{\scriptsize\boldmath$w$},\mbox{\scriptsize\boldmath$r$})},P_{Y|\mbox{\scriptsize\boldmath$X$}})^{1-s_1}
e^{-Ns_1\tau_{(\mbox{\tiny\boldmath$r$}, P_{Y|\mbox{\tiny\boldmath$X$}},\mathcal{S})}(\mbox{\scriptsize\boldmath$y$})} \right] \nonumber\\
&=&  \sum_{\mbox{\scriptsize\boldmath$y$}} E_{\mbox{\scriptsize\boldmath$\theta$}_{\mathcal{S}}} \left[  E_{\mbox{\scriptsize\boldmath$\theta$}_{\bar{\mathcal{S}}}} \left[ P(\mbox{\boldmath $y$}|\mbox{\boldmath $x$}_{(\mbox{\scriptsize\boldmath$w$},\mbox{\scriptsize\boldmath$r$})},P_{Y|\mbox{\scriptsize\boldmath$X$}})^{1-s_1} \right] e^{-Ns_1 \tau_{(\mbox{\tiny\boldmath$r$}, P_{Y|\mbox{\tiny\boldmath$X$}},\mathcal{S})}(\mbox{\scriptsize\boldmath$y$}) } \right].
\end{eqnarray}

{\bf Step III:} Upper-Bounding $P_{i[(\tilde{\mbox{\scriptsize\boldmath$r$}},\tilde{P}_{Y|\mbox{\tiny\boldmath$X$}}),(\mbox{\scriptsize\boldmath$r$},P_{Y|\mbox{\tiny\boldmath$X$}}),\mathcal{S}]}$

Given $\tilde{\mbox{\boldmath$r$}} \notin \mathcal{R}$ and $\mbox{\boldmath$r$} \in \mathcal{R}$, we rewrite $P_{i[(\tilde{\mbox{\scriptsize\boldmath$r$}},\tilde{P}_{Y|\mbox{\tiny\boldmath$X$}}),(\mbox{\scriptsize\boldmath$r$},P_{Y|\mbox{\tiny\boldmath$X$}}),\mathcal{S}]}$
as
\begin{eqnarray}
\label{P_iBound1}
P_{i[(\tilde{\mbox{\scriptsize\boldmath$r$}},\tilde{P}_{Y|\mbox{\tiny\boldmath$X$}}),(\mbox{\scriptsize\boldmath$r$},P_{Y|\mbox{\tiny\boldmath$X$}}),\mathcal{S}]}
= E_{\theta} \left[\sum_{\mbox{\scriptsize\boldmath$y$}} P(\mbox{\boldmath $y$}|\mbox{\boldmath $x$}_{(\tilde{\mbox{\scriptsize\boldmath$w$}},\tilde{\mbox{\scriptsize\boldmath$r$}})},\tilde{P}_{Y|\mbox{\scriptsize\boldmath$X$}}) \phi_{[(\tilde{\mbox{\scriptsize\boldmath$r$}},\tilde{P}_{Y|\mbox{\tiny\boldmath$X$}}),(\mbox{\scriptsize\boldmath$r$},P_{Y|\mbox{\tiny\boldmath$X$}}),\mathcal{S}]} (\mbox{\boldmath$y$}) \right],
\end{eqnarray}
where $\phi_{[(\tilde{\mbox{\scriptsize\boldmath$r$}},\tilde{P}_{Y|\mbox{\tiny\boldmath$X$}}),(\mbox{\scriptsize\boldmath$r$},P_{Y|\mbox{\tiny\boldmath$X$}}),\mathcal{S}]} (\mbox{\boldmath$y$}) = 1$ if there exists a triplet $(\mbox{\boldmath$w$},\mbox{\boldmath$r$},P_{Y|\mbox{\scriptsize\boldmath$X$}})$ with $(\mbox{\boldmath$r$}, P_{Y|\mbox{\scriptsize\boldmath$X$}}) \in \mathcal{R}$, $(\mbox{\boldmath$w$}_{\mathcal{S}},\mbox{\boldmath$r$}_{\mathcal{S}}) = (\tilde{\mbox{\boldmath$w$}}_{\mathcal{S}},\tilde{\mbox{\boldmath$r$}}_{\mathcal{S}})$, and $(w_k,r_k) \neq (\tilde{w}_k, \tilde{r}_k)$ for all $k\notin \mathcal{S}$, such that $P(\mbox{\boldmath $y$}|\mbox{\boldmath $x$}_{(\mbox{\scriptsize\boldmath$w$},\mbox{\scriptsize\boldmath$r$})},P_{Y|\mbox{\scriptsize\boldmath$X$}}) > e^{-N\tau_{(\mbox{\tiny\boldmath$r$}, P_{Y|\mbox{\tiny\boldmath$X$}},\mathcal{S})}(\mbox{\scriptsize\boldmath$y$})}$ is satisfied. Otherwise, $\phi_{[(\tilde{\mbox{\scriptsize\boldmath$r$}},\tilde{P}_{Y|\mbox{\tiny\boldmath$X$}}),(\mbox{\scriptsize\boldmath$r$},P_{Y|\mbox{\tiny\boldmath$X$}}),\mathcal{S}]} (\mbox{\boldmath$y$}) = 0$.

For any $s_2 >0$ and $\tilde{\rho} > 0$, $\phi_{[(\tilde{\mbox{\scriptsize\boldmath$r$}},\tilde{P}_{Y|\mbox{\tiny\boldmath$X$}}),(\mbox{\scriptsize\boldmath$r$},P_{Y|\mbox{\tiny\boldmath$X$}}),\mathcal{S}]} (\mbox{\boldmath$y$})$
 can be bounded by,
\begin{eqnarray}
\label{PhiBoundPiCC}
\phi_{[(\tilde{\mbox{\scriptsize\boldmath$r$}},\tilde{P}_{Y|\mbox{\tiny\boldmath$X$}}),(\mbox{\scriptsize\boldmath$r$},P_{Y|\mbox{\tiny\boldmath$X$}}),\mathcal{S}]} (\mbox{\boldmath$y$}) \le
\left[ \frac{\sum_{\mbox{\scriptsize\boldmath$w$},(\mbox{\scriptsize\boldmath$w$}_{\mathcal{S}},\mbox{\scriptsize\boldmath$r$}_{\mathcal{S}})=(\tilde{\mbox{\scriptsize\boldmath$w$}}_{\mathcal{S}},\tilde{\mbox{\scriptsize\boldmath$r$}}_{\mathcal{S}}),(w_k,r_k) \neq (\tilde{w}_k, \tilde{r}_k) \forall k \notin \mathcal{S}} P(\mbox{\boldmath $y$}|\mbox{\boldmath $x$}_{(\mbox{\scriptsize\boldmath$w$},\mbox{\scriptsize\boldmath$r$})},P_{Y|\mbox{\scriptsize\boldmath$X$}})^{\frac{s_2}{\tilde{\rho}}}}
{e^{-N \frac{s_2} {\tilde{\rho}} \tau_{(\mbox{\tiny\boldmath$r$}, P_{Y|\mbox{\tiny\boldmath$X$}},\mathcal{S})}(\mbox{\scriptsize\boldmath$y$})}}\right] ^{\tilde{\rho}}.
\end{eqnarray}
Substituting (\ref{PhiBoundPiCC}) into (\ref{P_iBound1}) yields,
\begin{eqnarray}
\label{P_iBound2}
P_{i[(\tilde{\mbox{\scriptsize\boldmath$r$}},\tilde{P}_{Y|\mbox{\tiny\boldmath$X$}}),(\mbox{\scriptsize\boldmath$r$},P_{Y|\mbox{\tiny\boldmath$X$}}),\mathcal{S}]}
&\le& \sum_{\mbox{\scriptsize\boldmath$y$}}  E_{\theta} \left[P(\mbox{\boldmath $y$}|\mbox{\boldmath $x$}_{(\tilde{\mbox{\scriptsize\boldmath$w$}},\tilde{\mbox{\scriptsize\boldmath$r$}})},\tilde{P}_{Y|\mbox{\scriptsize\boldmath$X$}})e^{N s_2\tau_{(\mbox{\tiny\boldmath$r$}, P_{Y|\mbox{\tiny\boldmath$X$}},\mathcal{S})}(\mbox{\scriptsize\boldmath$y$})} \right.\nonumber\\
&&\quad \left. \times \left[ \sum_{\mbox{\scriptsize\boldmath$w$},(\mbox{\scriptsize\boldmath$w$}_{\mathcal{S}},\mbox{\scriptsize\boldmath$r$}_{\mathcal{S}})=(\tilde{\mbox{\scriptsize\boldmath$w$}}_{\mathcal{S}},\tilde{\mbox{\scriptsize\boldmath$r$}}_{\mathcal{S}}) } P(\mbox{\boldmath $y$}|\mbox{\boldmath $x$}_{(\mbox{\scriptsize\boldmath$w$},\mbox{\scriptsize\boldmath$r$})},P_{Y|\mbox{\scriptsize\boldmath$X$}})^{\frac{s_2}{\tilde{\rho}}}
\right] ^{\tilde{\rho}} \right].
\end{eqnarray}
The independence between $(\mbox{\boldmath$w$}_{\bar{\mathcal{S}}},\mbox{\boldmath$r$}_{\bar{\mathcal{S}}})$ and $(\tilde{\mbox{\boldmath$w$}}_{\bar{\mathcal{S}}},\tilde{\mbox{\boldmath$r$}}_{\bar{\mathcal{S}}})$ allows us to rewrite the above bound as
\begin{eqnarray}
\label{P_iBound3}
P_{i[(\tilde{\mbox{\scriptsize\boldmath$r$}},\tilde{P}_{Y|\mbox{\tiny\boldmath$X$}}),(\mbox{\scriptsize\boldmath$r$},P_{Y|\mbox{\tiny\boldmath$X$}}),\mathcal{S}]}
&\le& \sum_{\mbox{\scriptsize\boldmath$y$}}  E_{\theta_{\mathcal{S}}} \left[ E_{\theta_{\bar{\mathcal{S}}}} \left[ P(\mbox{\boldmath $y$}|\mbox{\boldmath $x$}_{(\tilde{\mbox{\scriptsize\boldmath$w$}},\tilde{\mbox{\scriptsize\boldmath$r$}})},\tilde{P}_{Y|\mbox{\scriptsize\boldmath$X$}}) \right] e^{N s_2\tau_{(\mbox{\tiny\boldmath$r$}, P_{Y|\mbox{\tiny\boldmath$X$}},\mathcal{S})}(\mbox{\scriptsize\boldmath$y$})} \right.\nonumber\\
&&\quad \left. \times E_{\theta_{\bar{\mathcal{S}}}} \left[  \left[ \sum_{\mbox{\scriptsize\boldmath$w$},(\mbox{\scriptsize\boldmath$w$}_{\mathcal{S}},\mbox{\scriptsize\boldmath$r$}_{\mathcal{S}})=(\tilde{\mbox{\scriptsize\boldmath$w$}}_{\mathcal{S}},\tilde{\mbox{\scriptsize\boldmath$r$}}_{\mathcal{S}})} P(\mbox{\boldmath $y$}|\mbox{\boldmath $x$}_{(\mbox{\scriptsize\boldmath$w$},\mbox{\scriptsize\boldmath$r$})},P_{Y|\mbox{\scriptsize\boldmath$X$}})^{\frac{s_2}{\tilde{\rho}}}
\right] ^{\tilde{\rho}}\right] \right].
\end{eqnarray}

With the assumption of $0 < \tilde{\rho} \le 1$, the inequality in (\ref{P_iBound3}) becomes
\begin{eqnarray}
\label{MCBound4P_i}
P_{i[(\tilde{\mbox{\scriptsize\boldmath$r$}},\tilde{P}_{Y|\mbox{\tiny\boldmath$X$}}),(\mbox{\scriptsize\boldmath$r$},P_{Y|\mbox{\tiny\boldmath$X$}}),\mathcal{S}]} &\le& \sum_{\mbox{\scriptsize\boldmath$y$}} E_{\mbox{\scriptsize\boldmath$\theta$}_{\mathcal{S}}} \left[  E_{\mbox{\scriptsize\boldmath$\theta$}_{\bar{\mathcal{S}}}} \left[ P(\mbox{\boldmath$y$}|\mbox{\boldmath$x$}_{(\tilde{\mbox{\scriptsize\boldmath$w$}},\tilde{\mbox{\scriptsize\boldmath$r$}})},\tilde{P}_{Y|\mbox{\scriptsize\boldmath$X$}}) \right] \right. \nonumber\\
&&\left. \times E_{\mbox{\scriptsize\boldmath$\theta$}_{\bar{\mathcal{S}}}} \left\{ \left[ P(\mbox{\boldmath$y$}|\mbox{\boldmath$x$}_{(\mbox{\scriptsize\boldmath$w$},\mbox{\scriptsize\boldmath$r$})},P_{Y|\mbox{\scriptsize\boldmath$X$}})^{\frac{s_2}{\tilde{\rho}}} \right] \right\}^{\tilde{\rho}} e^{Ns_2 \tau_{(\mbox{\tiny\boldmath$r$}, P_{Y|\mbox{\tiny\boldmath$X$}},\mathcal{S})}(\mbox{\scriptsize\boldmath$y$}) } e^{N\tilde{\rho}\sum_{k\notin\mathcal{S}}r_k} \right]  \nonumber\\
&\le& \max_{(\mbox{\scriptsize\boldmath$r$}', P'_{Y|\mbox{\tiny\boldmath$X$}}) \notin \mathcal{R},\mbox{\scriptsize\boldmath$r$}'_{\mathcal{S}} = \mbox{\scriptsize\boldmath$r$}_{\mathcal{S}}} \sum_{\mbox{\scriptsize\boldmath$y$}} E_{\mbox{\scriptsize\boldmath$\theta$}_{\mathcal{S}}} \left[  E_{\mbox{\scriptsize\boldmath$\theta$}_{\bar{\mathcal{S}}}} \left[ P(\mbox{\boldmath$y$}|\mbox{\boldmath$x$}_{(\mbox{\scriptsize\boldmath$w$}',\mbox{\scriptsize\boldmath$r$}')},P'_{Y|\mbox{\scriptsize\boldmath$X$}}) \right] \right. \nonumber\\
&&\left. \times E_{\mbox{\scriptsize\boldmath$\theta$}_{\bar{\mathcal{S}}}} \left\{ \left[ P(\mbox{\boldmath$y$}|\mbox{\boldmath$x$}_{(\mbox{\scriptsize\boldmath$w$},\mbox{\scriptsize\boldmath$r$})},P_{Y|\mbox{\scriptsize\boldmath$X$}})^{\frac{s_2}{\tilde{\rho}}} \right] \right\}^{\tilde{\rho}} e^{Ns_2 \tau_{(\mbox{\tiny\boldmath$r$}, P_{Y|\mbox{\tiny\boldmath$X$}},\mathcal{S})}(\mbox{\scriptsize\boldmath$y$}) } e^{N\tilde{\rho}\sum_{k\notin\mathcal{S}}r_k} \right].
\end{eqnarray}
Note that the upper bound in (\ref{MCBound4P_i}) is no longer a function of $(\tilde{\mbox{\boldmath$r$}}_{\bar{\mathcal{S}}},\tilde{P}_{Y|\mbox{\scriptsize\boldmath$X$}})$.

{\bf Step IV: } Choosing $\tau_{(\mbox{\tiny\boldmath$r$}, P_{Y|\mbox{\tiny\boldmath$X$}},\mathcal{S})}(\mbox{\boldmath$y$})$

The value of $\tau_{(\mbox{\tiny\boldmath$r$}, P_{Y|\mbox{\tiny\boldmath$X$}},\mathcal{S})}(\mbox{\boldmath$y$})$ can be determined by jointly optimizing the bounds in (\ref{MCBound4P_t}) and (\ref{MCBound4P_i}). Consequently, given $(\mbox{\boldmath$r$}, P_{Y|\mbox{\scriptsize\boldmath$X$}}) \in \mathcal{R}$, $\mbox{\boldmath$y$}$ and auxiliary variables $s_1>0$, $s_2>0$, $0< \tilde{\rho} \le 1$, we choose  $\tau_{(\mbox{\tiny\boldmath$r$}, P_{Y|\mbox{\tiny\boldmath$X$}},\mathcal{S})}(\mbox{\boldmath$y$})$ such that the following equality is satisfied,
\begin{eqnarray}
\label{TauEquality}
&& E_{\mbox{\scriptsize\boldmath$\theta$}_{\bar{\mathcal{S}}}} \left[ P(\mbox{\boldmath $y$}|\mbox{\boldmath $x$}_{(\mbox{\scriptsize\boldmath$w$},\mbox{\scriptsize\boldmath$r$})},P_{Y|\mbox{\scriptsize\boldmath$X$}})^{1-s_1} \right] e^{-Ns_1 \tau_{(\mbox{\tiny\boldmath$r$}, P_{Y|\mbox{\tiny\boldmath$X$}},\mathcal{S})}(\mbox{\scriptsize\boldmath$y$}) } \nonumber\\
&& = E_{\mbox{\scriptsize\boldmath$\theta$}_{\bar{\mathcal{S}}}} \left[ P(\mbox{\boldmath$y$}|\mbox{\boldmath$x$}_{(\tilde{\mbox{\scriptsize\boldmath$w$}}^*,\tilde{\mbox{\scriptsize\boldmath$r$}}^*)},\tilde{P}^*_{Y|\mbox{\scriptsize\boldmath$X$}}) \right]\times E_{\mbox{\scriptsize\boldmath$\theta$}_{\bar{\mathcal{S}}}} \left\{ \left[ P(\mbox{\boldmath$y$}|\mbox{\boldmath$x$}_{(\mbox{\scriptsize\boldmath$w$},\mbox{\scriptsize\boldmath$r$})},P_{Y|\mbox{\scriptsize\boldmath$X$}})^{\frac{s_2}{\tilde{\rho}}} \right] \right\}^{\tilde{\rho}} e^{Ns_2 \tau_{(\mbox{\tiny\boldmath$r$}, P_{Y|\mbox{\tiny\boldmath$X$}},\mathcal{S})}(\mbox{\scriptsize\boldmath$y$}) } e^{N\tilde{\rho}\sum_{k\notin\mathcal{S}}r_k} .
\end{eqnarray}
where $(\tilde{\mbox{\boldmath$r$}}^* ,\tilde{P}_{Y|\mbox{\scriptsize\boldmath$X$}}^*)$ is defined as\footnote{Although the notation of $\tilde{\mbox{\boldmath$w$}}^*$ is used in (\ref{TauEquality}), the result is actually invariant to any choice of the message vector.}
\begin{eqnarray}
\label{DefineR}
(\tilde{\mbox{\boldmath$r$}}^* ,\tilde{P}_{Y|\mbox{\scriptsize\boldmath$X$}}^*) &= & \mathop{\mbox{argmax}}_{(\mbox{\scriptsize\boldmath$r$}', P'_{Y|\mbox{\tiny\boldmath$X$}}) \notin \mathcal{R},\mbox{\scriptsize\boldmath$r$}'_{\mathcal{S}} = \mbox{\scriptsize\boldmath$r$}_{\mathcal{S}}} \sum_{\mbox{\scriptsize\boldmath$y$}} E_{\mbox{\scriptsize\boldmath$\theta$}_{\mathcal{S}}} \left[  E_{\mbox{\scriptsize\boldmath$\theta$}_{\bar{\mathcal{S}}}} \left[ P(\mbox{\boldmath$y$}|\mbox{\boldmath$x$}_{(\mbox{\scriptsize\boldmath$w$}',\mbox{\scriptsize\boldmath$r$}')},P'_{Y|\mbox{\scriptsize\boldmath$X$}}) \right] \right. \nonumber\\
&& \left. \times E_{\mbox{\scriptsize\boldmath$\theta$}_{\bar{\mathcal{S}}}} \left\{ \left[ P(\mbox{\boldmath$y$}|\mbox{\boldmath$x$}_{(\mbox{\scriptsize\boldmath$w$},\mbox{\scriptsize\boldmath$r$})},P_{Y|\mbox{\scriptsize\boldmath$X$}})^{\frac{s_2}{\tilde{\rho}}} \right] \right\}^{\tilde{\rho}} e^{Ns_2 \tau_{(\mbox{\tiny\boldmath$r$}, P_{Y|\mbox{\tiny\boldmath$X$}},\mathcal{S})}(\mbox{\scriptsize\boldmath$y$}) } e^{N\tilde{\rho}\sum_{k\notin\mathcal{S}}r_k} \right].
\end{eqnarray}
Finding a solution for (\ref{TauEquality}) is always possible since that the left hand side of (\ref{TauEquality}) decreases with $\tau_{(\mbox{\tiny\boldmath$r$}, P_{Y|\mbox{\tiny\boldmath$X$}},\mathcal{S})}(\mbox{\boldmath$y$})$, while the right hand side of (\ref{TauEquality}) increases with $\tau_{(\mbox{\tiny\boldmath$r$}, P_{Y|\mbox{\tiny\boldmath$X$}},\mathcal{S})}(\mbox{\boldmath$y$})$. This yields the desired typicality threshold, denoted by $\tau_{(\mbox{\scriptsize\boldmath$r$},P_{Y|\mbox{\tiny\boldmath$X$}},\mathcal{S})}^*(\mbox{\scriptsize\boldmath$y$})$, which gives
\begin{eqnarray}
\label{TauExpression}
&& e^{-N\tau_{(\mbox{\scriptsize\boldmath$r$},P_{Y|\mbox{\tiny\boldmath$X$}},\mathcal{S})}^*(\mbox{\scriptsize\boldmath$y$})} =
\frac{\left\{E_{\mbox{\scriptsize\boldmath$\theta$}_{\bar{\mathcal{S}}}} \left[ P(\mbox{\boldmath$y$}|\mbox{\boldmath$x$}_{(\tilde{\mbox{\scriptsize\boldmath$w$}}^*,\tilde{\mbox{\scriptsize\boldmath$r$}}^*)},\tilde{P}^*_{Y|\mbox{\scriptsize\boldmath$X$}}) \right]\right\}^{\frac{1}{s_1+s_2}}E_{\mbox{\scriptsize\boldmath$\theta$}_{\bar{\mathcal{S}}}} \left\{ \left[ P(\mbox{\boldmath$y$}|\mbox{\boldmath$x$}_{(\mbox{\scriptsize\boldmath$w$},\mbox{\scriptsize\boldmath$r$})},P_{Y|\mbox{\scriptsize\boldmath$X$}})^{\frac{s_2}{\tilde{\rho}}} \right] \right\}^{\frac{\tilde{\rho}}{s_1+s_2}} }         {\left\{  E_{\mbox{\scriptsize\boldmath$\theta$}_{\bar{\mathcal{S}}}} \left[ P(\mbox{\boldmath $y$}|\mbox{\boldmath $x$}_{(\mbox{\scriptsize\boldmath$w$},\mbox{\scriptsize\boldmath$r$})},P_{Y|\mbox{\scriptsize\boldmath$X$}})^{1-s_1} \right] \right\} ^{\frac{1}{s_1+s_2}}}\nonumber\\
&& \qquad\qquad\qquad\qquad\qquad\qquad\qquad\qquad\qquad\qquad\qquad\qquad\qquad \times e^{N\frac{\tilde{\rho}}{s_1+s_2}\sum_{k\notin\mathcal{S}}r_k}.
\end{eqnarray}

Substituting (\ref{TauExpression}) into (\ref{MCBound4P_t}), we get
\begin{eqnarray}
\label{MCBound4P_tA}
P_{t[\mbox{\scriptsize\boldmath$r$},P_{Y|\mbox{\tiny\boldmath$X$}},\mathcal{S}]} &\le & \sum_{\mbox{\scriptsize\boldmath$y$}} E_{\mbox{\scriptsize\boldmath$\theta$}_{\mathcal{S}}} \left[  E_{\mbox{\scriptsize\boldmath$\theta$}_{\bar{\mathcal{S}}}} \left[ P(\mbox{\boldmath $y$}|\mbox{\boldmath $x$}_{(\mbox{\scriptsize\boldmath$w$},\mbox{\scriptsize\boldmath$r$})},P_{Y|\mbox{\scriptsize\boldmath$X$}})^{1-s_1} \right]^{\frac{s_2}{s_1+s_2}} \left\{E_{\mbox{\scriptsize\boldmath$\theta$}_{\bar{\mathcal{S}}}} \left[ P(\mbox{\boldmath$y$}|\mbox{\boldmath$x$}_{(\tilde{\mbox{\scriptsize\boldmath$w$}}^*,\tilde{\mbox{\scriptsize\boldmath$r$}}^*)},\tilde{P}^*_{Y|\mbox{\scriptsize\boldmath$X$}}) \right]\right\}^{\frac{s_1}{s_1+s_2}}\right.\nonumber\\
&& \left.\times E_{\mbox{\scriptsize\boldmath$\theta$}_{\bar{\mathcal{S}}}} \left\{ \left[ P(\mbox{\boldmath$y$}|\mbox{\boldmath$x$}_{(\mbox{\scriptsize\boldmath$w$},\mbox{\scriptsize\boldmath$r$})},P_{Y|\mbox{\scriptsize\boldmath$X$}})^{\frac{s_2}{\tilde{\rho}}} \right] \right\}^{\frac{s_1\tilde{\rho}}{s_1+s_2}}  e^{N\frac{s_1\tilde{\rho}}{s_1+s_2}\sum_{k\notin\mathcal{S}}r_k} \right].
\end{eqnarray}
Let $s_2 < \tilde{\rho}$ and $s_1 = 1 - \frac{s_2}{\tilde{\rho}}$, and then do a variable change with $\rho = \frac{\tilde{\rho}(\tilde{\rho}-s_2)}{\tilde{\rho} - (1-\tilde{\rho})s_2}$ and $s = 1- \frac{\tilde{\rho}-s_2}{\tilde{\rho} - (1-\tilde{\rho})s_2}$. Consequently, inequality (\ref{MCBound4P_tA}) becomes,
\begin{eqnarray}
\label{MCBound4P_tB}
&&P_{t[\mbox{\scriptsize\boldmath$r$},P_{Y|\mbox{\tiny\boldmath$X$}},\mathcal{S}]}  \le  \max_{(\mbox{\scriptsize\boldmath$r$}', P'_{Y|\mbox{\tiny\boldmath$X$}}) \notin \mathcal{R},\mbox{\scriptsize\boldmath$r$}'_{\mathcal{S}} = \mbox{\scriptsize\boldmath$r$}_{\mathcal{S}}} e^{N\rho \sum_{k\not\in \mathcal{S}}r_k} \left\{ \sum_Y \sum_{\mbox{\scriptsize \boldmath $X$}_{\mathcal{S}}} \prod_{k\in \mathcal{S}} P_{X|r_k}(X_k)\right.               \nonumber \\
&&\qquad\qquad \left. \times \left(\sum_{\mbox{\scriptsize \boldmath $X$}_{\bar{\mathcal{S}}}}\prod_{k \not\in \mathcal{S}}P_{X|r_k}(X_k)P_{Y|\mbox{\scriptsize\boldmath$X$}}(Y|\mbox{\boldmath $X$})^{\frac{s}{s+\rho}} \right)^{s+\rho}\left(\sum_{\mbox{\scriptsize \boldmath $X$}_{\bar{\mathcal{S}}}}\prod_{k \not\in \mathcal{S}}P_{X|r'_k}(X_k)P'_{Y|\mbox{\scriptsize\boldmath$X$}}(Y|\mbox{\boldmath $X$})\right)^{1-s}\right\}^N.
\end{eqnarray}
Similarly, we can obtain the same upper bound for $P_{i[(\tilde{\mbox{\scriptsize\boldmath$r$}},\tilde{P}_{Y|\mbox{\tiny\boldmath$X$}}),(\mbox{\scriptsize\boldmath$r$},P_{Y|\mbox{\tiny\boldmath$X$}}),\mathcal{S}]}$ as given at the right hand side of (\ref{MCBound4P_tB}). Since (\ref{MCBound4P_tB}) holds for all $0<\rho\le 1$ and $0<s\le 1-\rho$, we have
\begin{eqnarray}
\label{BoundPiPt}
P_{t[\mbox{\scriptsize\boldmath$r$},P_{Y|\mbox{\tiny\boldmath$X$}},\mathcal{S}]}, P_{i[(\tilde{\mbox{\scriptsize\boldmath$r$}},\tilde{P}_{Y|\mbox{\tiny\boldmath$X$}}),(\mbox{\scriptsize\boldmath$r$},P_{Y|\mbox{\tiny\boldmath$X$}}),\mathcal{S}]} \le \max_{(\mbox{\scriptsize\boldmath$r$}', P'_{Y|\mbox{\tiny\boldmath$X$}}) \notin \mathcal{R},\mbox{\scriptsize\boldmath$r$}'_{\mathcal{S}} = \mbox{\scriptsize\boldmath$r$}_{\mathcal{S}}} \exp \left\{ -NE_i(\mathcal{S}, \mbox{\boldmath$r$},\mbox{\boldmath $r$}', P_{Y|\mbox{\scriptsize\boldmath$X$}}, P'_{Y|\mbox{\scriptsize\boldmath$X$}} ) \right\},
\end{eqnarray}
where $E_i(\mathcal{S}, \mbox{\boldmath$r$},\mbox{\boldmath $r$}', P_{Y|\mbox{\scriptsize\boldmath$X$}}, P'_{Y|\mbox{\scriptsize\boldmath$X$}} ) $ is given in (\ref{EmEiMultiMC}).

By substituting (\ref{PmBound4}) and (\ref{BoundPiPt}) into (\ref{ProofPes}), we get the desired result.

\subsection{Proof of Theorem \ref{TheoremContinuousChannel}}
\label{AppendixTheoremContinuousChannel}
We assume that the following decoding algorithm is used at the receiver. Given the channel output sequence $\mbox{\boldmath$y$}$, the receiver outputs a message and rate vector pair $(\mbox{\boldmath$w$},\mbox{\boldmath$r$})$ together with a channel class $\mathcal{F}$ such that $(\mbox{\boldmath$r$},\mathcal{F}) \in \mathcal{R}$ if for all user subset $\mathcal{S} \subset \{1,\cdots,K\}$, the following condition is satisfied,
\begin{eqnarray}
\label{RevisedCriterionCC}
&& -\frac{1}{N} \log Pr \left\{ \mbox{\boldmath$y$} | \mbox{\boldmath$x$}_{(\mbox{\scriptsize\boldmath$w$},\mbox{\scriptsize\boldmath$r$})},P^{\mathcal{F}}_{\min} \right\} <
-\frac{1}{N} \log Pr \left\{ \mbox{\boldmath$y$} | \mbox{\boldmath$x$}_{(\tilde{\mbox{\scriptsize\boldmath$w$}},\tilde{\mbox{\scriptsize\boldmath$r$}})},P^{\tilde{\mathcal{F}}}_{\max} \right\},   \nonumber\\
&& \quad \mbox{ for all } (\tilde{\mbox{\boldmath$w$}},\tilde{\mbox{\boldmath$r$}},\tilde{\mathcal{F}}), (\tilde{\mbox{\boldmath$w$}}_{\mathcal{S}},\tilde{\mbox{\boldmath$r$}}_{\mathcal{S}}) = (\mbox{\boldmath$w$}_{\mathcal{S}}, \mbox{\boldmath$r$}_{\mathcal{S}}),(\tilde{w}_{k},\tilde{r}_k)\neq (w_k,r_k), \forall k \notin \mathcal{S}, \nonumber\\
&& \quad\mbox{ and }(\mbox{\boldmath$w$},\mbox{\boldmath$r$},P^{\mathcal{F}}_{\min}), (\tilde{\mbox{\boldmath$w$}},\tilde{\mbox{\boldmath$r$}},P^{\tilde{\mathcal{F}}}_{\max}) \in \mathcal{R}_{(\mathcal{S},\mbox{\scriptsize\boldmath$y$})},  \mbox{ with }\nonumber\\
&& \mathcal{R}_{(\mathcal{S},\mbox{\scriptsize\boldmath$y$})} = \left\{ (\tilde{\mbox{\boldmath$w$}},\tilde{\mbox{\boldmath$r$}},P^{\tilde{\mathcal{F}}}) | (\tilde{\mbox{\boldmath$r$}},\tilde{\mathcal{F}}) \in \mathcal{R}, -\frac{1}{N} Pr \left\{ \mbox{\boldmath$y$} | \mbox{\boldmath$x$}_{(\tilde{\mbox{\scriptsize\boldmath$w$}},\tilde{\mbox{\scriptsize\boldmath$r$}})},P^{\tilde{\mathcal{F}}} \right\} < \tau_{(\tilde{\mbox{\scriptsize\boldmath$r$}},P^{\tilde{\mathcal{F}}},\mathcal{S})}(\mbox{\boldmath$y$}) \right\},
\end{eqnarray}
where $\tau_{(\tilde{\mbox{\scriptsize\boldmath$r$}}, P^{\tilde{\mathcal{F}}},\mathcal{S})}(\cdot)$ is the typicality threshold function. Again, we will first analyze the error performance for each individual $\mathcal{S}$ and then derive the overall error performance by taking the union over all $\mathcal{S}$.

For a given user subset $\mathcal{S} \subset \{1,\cdots,K\}$, the following probability terms are defined.

First, assume that $(\mbox{\boldmath$w$},\mbox{\boldmath$r$})$ is transmitted over channel $P_{Y|\mbox{\scriptsize\boldmath$X$}} \in \mathcal{F}$, with $(\mbox{\boldmath$r$},\mathcal{F}) \in \mathcal{R}$. Let $P_{t[\mbox{\scriptsize\boldmath$r$}, \mathcal{F}, P_{Y|\mbox{\tiny\boldmath$X$}}, \mathcal{S}]}$ be the probability that the likelihood value of the transmitted codeword vector calculated using $P^{\mathcal{F}}_{\min}$ is no larger than the corresponding typicality threshold,
\begin{eqnarray}
\label{ProofPtMC}
P_{t[\mbox{\scriptsize\boldmath$r$},\mathcal{F},P_{Y|\mbox{\tiny\boldmath$X$}},\mathcal{S}]} = Pr \left\{ P(\mbox{\boldmath $y$}|\mbox{\boldmath $x$}_{(\mbox{\scriptsize\boldmath$w$},\mbox{\scriptsize\boldmath$r$})},P^{\mathcal{F}}_{\min}) \le e^{-N\tau_{(\mbox{\tiny\boldmath$r$}, P^{\mathcal{F}}_{\min},\mathcal{S})}(\mbox{\scriptsize\boldmath$y$})}\right\}.
\end{eqnarray}
Define $P_{m[(\mbox{\scriptsize\boldmath$r$},\mathcal{F}),(\tilde{\mbox{\scriptsize\boldmath$r$}},\tilde{\mathcal{F}}), P_{Y|\mbox{\tiny\boldmath$X$}}, \mathcal{S}]}$ as the probability that the likelihood value of the transmitted codeword vector calculated using $P^{\mathcal{F}}_{\min}$ is no larger than that of another codeword $(\mbox{\boldmath$w$},\mbox{\boldmath$r$})$ with $(\tilde{\mbox{\boldmath$w$}}_{\mathcal{S}},\tilde{\mbox{\boldmath$r$}}_{\mathcal{S}}) = (\mbox{\boldmath$w$}_{\mathcal{S}}, \mbox{\boldmath$r$}_{\mathcal{S}}),(\tilde{w}_{k},\tilde{r}_{k})\neq (w_k,r_k), \forall k \notin \mathcal{S}$, calculated using $P^{\tilde{\mathcal{F}}}_{\max}$ with $(\tilde{\mbox{\boldmath$r$}},\tilde{\mathcal{F}}) \in \mathcal{R}$,
\begin{eqnarray}
\label{ProofPmCC}
&& P_{m[(\mbox{\scriptsize\boldmath$r$},\mathcal{F}),(\tilde{\mbox{\scriptsize\boldmath$r$}},\tilde{\mathcal{F}}), P_{Y|\mbox{\tiny\boldmath$X$}}, \mathcal{S}]} = Pr \left\{ P(\mbox{\boldmath $y$}|\mbox{\boldmath $x$}_{(\mbox{\scriptsize\boldmath$w$},\mbox{\scriptsize\boldmath$r$})},P^{\mathcal{F}}_{\min}) \le P(\mbox{\boldmath $y$}|\mbox{\boldmath $x$}_{(\tilde{\mbox{\scriptsize\boldmath$w$}}, \tilde{\mbox{\scriptsize\boldmath$r$}})},P^{\tilde{\mathcal{F}}}_{\max})\right\}\nonumber\\
&& \quad (\tilde{\mbox{\boldmath$w$}},\tilde{\mbox{\boldmath$r$}},\tilde{\mathcal{F}}),(\tilde{\mbox{\boldmath$r$}},\tilde{\mathcal{F}}) \in \mathcal{R}, (\tilde{\mbox{\boldmath$w$}}_{\mathcal{S}},\tilde{\mbox{\boldmath$r$}}_{\mathcal{S}}) = (\mbox{\boldmath$w$}_{\mathcal{S}}, \mbox{\boldmath$r$}_{\mathcal{S}}),(\tilde{w}_{k},\tilde{r}_k)\neq (w_k,r_k), \forall k \notin \mathcal{S}.
\end{eqnarray}

Second, assume that $(\tilde{\mbox{\boldmath$w$}},\tilde{\mbox{\boldmath$r$}})$ is transmitted over channel $\tilde{P}_{Y|\mbox{\scriptsize\boldmath$X$}} \in \tilde{\mathcal{F}}$, with $(\tilde{\mbox{\boldmath$r$}},\tilde{\mathcal{F}}) \notin \mathcal{R}$. Define $P_{i[(\tilde{\mbox{\scriptsize\boldmath$r$}},\tilde{\mathcal{F}}),(\mbox{\scriptsize\boldmath$r$},\mathcal{F}), \tilde{P}_{Y|\mbox{\tiny\boldmath$X$}}, \mathcal{S}]}$  as the probability that the decoder finds a codeword $(\mbox{\boldmath$w$},\mbox{\boldmath$r$})$ with $ (\mbox{\boldmath$w$}_{\mathcal{S}}, \mbox{\boldmath$r$}_{\mathcal{S}})= (\tilde{\mbox{\boldmath$w$}}_{\mathcal{S}},\tilde{\mbox{\boldmath$r$}}_{\mathcal{S}}), (w_k,r_k) \neq (\tilde{w}_{k},\tilde{r}_k), \forall k \notin \mathcal{S} $, over channel class $\mathcal{F}$ with $(\mbox{\boldmath$r$}, \mathcal{F}) \in \mathcal{R}$, such that its likelihood value calculated using $P^{\mathcal{F}}_{\min}$ is larger than the corresponding typicality threshold,
\begin{eqnarray}
\label{ProofPiMC}
&& P_{i[(\tilde{\mbox{\scriptsize\boldmath$r$}},\tilde{\mathcal{F}}),(\mbox{\scriptsize\boldmath$r$},\mathcal{F}),\tilde{P}_{Y|\mbox{\tiny\boldmath$X$}}, \mathcal{S}]} = Pr \left\{ P(\mbox{\boldmath $y$}|\mbox{\boldmath $x$}_{(\mbox{\scriptsize\boldmath$w$},\mbox{\scriptsize\boldmath$r$})},P^{\mathcal{F}}_{\min}) > e^{-N\tau_{(\mbox{\tiny\boldmath$r$}, P^{\mathcal{F}}_{\min},\mathcal{S})}(\mbox{\scriptsize\boldmath$y$})}\right\},\nonumber\\
&& \quad (\mbox{\boldmath$w$},\mbox{\boldmath$r$}, \mathcal{F}),(\mbox{\boldmath$r$}, \mathcal{F}) \in \mathcal{R}, (\mbox{\boldmath$w$}_{\mathcal{S}}, \mbox{\boldmath$r$}_{\mathcal{S}})= (\tilde{\mbox{\boldmath$w$}}_{\mathcal{S}},\tilde{\mbox{\boldmath$r$}}_{\mathcal{S}}), (w_k,r_k) \neq (\tilde{w}_{k},\tilde{r}_k), \forall k \notin \mathcal{S}.
\end{eqnarray}

Consequently, the system error probability $P_{es}$ can be upper-bounded using the above probabilities terms as follows,
\begin{eqnarray}
\label{ProofPesCC}
P_{es} \le \max \left\{ \begin{array}{l} \max_{(\mbox{\scriptsize\boldmath$r$},P_{Y|\mbox{\tiny\boldmath$X$}}): P_{Y|\mbox{\tiny\boldmath$X$}} \in \mathcal{F}, (\mbox{\scriptsize\boldmath$r$},\mathcal{F}) \in \mathcal{R}} \sum_{\mathcal{S} \subset \{1,\cdots,K\}} \left[ P_{t[\mbox{\scriptsize\boldmath$r$}, \mathcal{F}, P_{Y|\mbox{\tiny\boldmath$X$}}, \mathcal{S}]} \right. \\
\left. \qquad\qquad\qquad\qquad\qquad\qquad + \sum_{(\tilde{\mbox{\scriptsize\boldmath$r$}},\tilde{\mathcal{F}})  \in \mathcal{R},\tilde{\mbox{\scriptsize\boldmath$r$}}_{\mathcal{S}} = \mbox{\scriptsize\boldmath$r$}_{\mathcal{S}}} P_{m[(\mbox{\scriptsize\boldmath$r$},\mathcal{F}),(\tilde{\mbox{\scriptsize\boldmath$r$}},\tilde{\mathcal{F}}), P_{Y|\mbox{\tiny\boldmath$X$}}, \mathcal{S}]} \right] , \\
\max_{(\tilde{\mbox{\scriptsize\boldmath$r$}},\tilde{P}_{Y|\mbox{\tiny\boldmath$X$}} ): \tilde{P}_{Y|\mbox{\tiny\boldmath$X$}} \in \mathcal{F}, (\tilde{\mbox{\scriptsize\boldmath$r$}},\tilde{\mathcal{F}})  \notin \mathcal{R}} \sum_{\mathcal{S} \subset \{1,\cdots,K\}}
\left[ \sum_{(\mbox{\scriptsize\boldmath$r$},\mathcal{F}) \in \mathcal{R},\mbox{\scriptsize\boldmath$r$}_{\mathcal{S}} = \tilde{\mbox{\scriptsize\boldmath$r$}}_{\mathcal{S}}} P_{i[(\tilde{\mbox{\scriptsize\boldmath$r$}},\tilde{\mathcal{F}}),(\mbox{\scriptsize\boldmath$r$},\mathcal{F}), \tilde{P}_{Y|\mbox{\tiny\boldmath$X$}}, \mathcal{S}]} \right] \end{array}\right\}.
\end{eqnarray}
Note that we have used the union bound over all user subsets $\mathcal{S}$ to obtain the probability bound in (\ref{ProofPesCC}). Next, we will derive individual bound for each of the probability terms on the right hand side of (\ref{ProofPesCC}).

A derivation similar to (\ref{PmBound1})-(\ref{PmBound4}) in Appendix \ref{AppendixTheoremMC} gives the upper bound on $P_{m[(\mbox{\scriptsize\boldmath$r$}, \mathcal{F}),(\tilde{\mbox{\scriptsize\boldmath$r$}},\tilde{\mathcal{F}}), P_{Y|\mbox{\tiny\boldmath$X$}}, \mathcal{S}]}$ as,
\begin{eqnarray}
\label{CCBound4P_m}
P_{m[(\mbox{\scriptsize\boldmath$r$},\mathcal{F}),(\tilde{\mbox{\scriptsize\boldmath$r$}},\tilde{\mathcal{F}}), P_{Y|\mbox{\tiny\boldmath$X$}}, \mathcal{S}]} &\le &
e^{N\rho \sum_{k\not\in \mathcal{S}}\tilde{r}_k} \sum_{\mbox{\scriptsize\boldmath$y$}} E_{\mbox{\scriptsize\boldmath$\theta$}_{\mathcal{S}}}
\left[ E_{\mbox{\scriptsize\boldmath$\theta$}_{\bar{\mathcal{S}}}} \left[ P(\mbox{\boldmath$y$} | \mbox{\boldmath$x$}_{(\mbox{\scriptsize\boldmath$w$},\mbox{\scriptsize\boldmath$r$})}, P_{Y|\mbox{\scriptsize\boldmath$X$}}) P(\mbox{\boldmath$y$} | \mbox{\boldmath$x$}_{(\mbox{\scriptsize\boldmath$w$},\mbox{\scriptsize\boldmath$r$})}, P^{\mathcal{F}}_{\min})^{-s} \right] \right. \nonumber\\
&& \left. \times \left[ E_{\mbox{\scriptsize\boldmath$\theta$}_{\bar{\mathcal{S}}}} \left[ P(\mbox{\boldmath$y$} | \mbox{\boldmath$x$}_{(\tilde{\mbox{\scriptsize\boldmath$w$}},\tilde{\mbox{\scriptsize\boldmath$r$}})}, P^{\tilde{\mathcal{F}}}_{\max})^{\frac{s}{\rho}} \right] \right] ^{\rho} \right]                 \nonumber\\
&\le& \exp \left\{ -NE_m(\mathcal{S}, \mbox{\boldmath$r$},\tilde{\mbox{\boldmath $r$}}, \mathcal{F}, \tilde{\mathcal{F}} ) \right\},
\end{eqnarray}
where $E_m(\mathcal{S}, \mbox{\boldmath$r$},\tilde{\mbox{\boldmath $r$}}, \mathcal{F}, \tilde{\mathcal{F}})$ is given in (\ref{EmEiCC}). Note that the second inequality in (\ref{CCBound4P_m}) is due to the fact that $P(\mbox{\boldmath$y$} | \mbox{\boldmath$x$}_{(\mbox{\scriptsize\boldmath$w$},\mbox{\scriptsize\boldmath$r$})}, P_{Y|\mbox{\scriptsize\boldmath$X$}}) \le P(\mbox{\boldmath$y$} | \mbox{\boldmath$x$}_{(\mbox{\scriptsize\boldmath$w$},\mbox{\scriptsize\boldmath$r$})}, P^{\mathcal{F}}_{\max}) $, and the right hand side of (\ref{CCBound4P_m}) is not a function of $P_{Y|\mbox{\scriptsize\boldmath$X$}}$.

Similarly, by using the same bounding techniques as in (\ref{PtBound1})-(\ref{MCBound4P_t}) and (\ref{P_iBound1})-(\ref{MCBound4P_i}) in Appendix \ref{AppendixTheoremMC}, we can upper bound $P_{t[\mbox{\scriptsize\boldmath$r$},\mathcal{F}, P_{Y|\mbox{\tiny\boldmath$X$}}, \mathcal{S}]} $ for any $s_1 >0$ by,
\begin{eqnarray}
\label{CCBound4P_t}
P_{t[\mbox{\scriptsize\boldmath$r$},\mathcal{F}, P_{Y|\mbox{\tiny\boldmath$X$}}, \mathcal{S}]} &\le & \sum_{\mbox{\scriptsize\boldmath$y$}} E_{\mbox{\scriptsize\boldmath$\theta$}_{\mathcal{S}}} \left[  E_{\mbox{\scriptsize\boldmath$\theta$}_{\bar{\mathcal{S}}}} \left[ P(\mbox{\boldmath $y$}|\mbox{\boldmath $x$}_{(\mbox{\scriptsize\boldmath$w$},\mbox{\scriptsize\boldmath$r$})},P_{Y|\mbox{\scriptsize\boldmath$X$}}) P(\mbox{\boldmath $y$}|\mbox{\boldmath $x$}_{(\mbox{\scriptsize\boldmath$w$},\mbox{\scriptsize\boldmath$r$})},P^{\mathcal{F}}_{\min})^{-s_1} \right] e^{-Ns_1 \tau_{(\mbox{\tiny\boldmath$r$}, P^{\mathcal{F}}_{\min} ,\mathcal{S})}(\mbox{\scriptsize\boldmath$y$}) } \right]            \nonumber\\
&\le& \sum_{\mbox{\scriptsize\boldmath$y$}} E_{\mbox{\scriptsize\boldmath$\theta$}_{\mathcal{S}}} \left[  E_{\mbox{\scriptsize\boldmath$\theta$}_{\bar{\mathcal{S}}}} \left[ P(\mbox{\boldmath $y$}|\mbox{\boldmath $x$}_{(\mbox{\scriptsize\boldmath$w$},\mbox{\scriptsize\boldmath$r$})},P^{\mathcal{F}}_{\max}) P(\mbox{\boldmath $y$}|\mbox{\boldmath $x$}_{(\mbox{\scriptsize\boldmath$w$},\mbox{\scriptsize\boldmath$r$})},P^{\mathcal{F}}_{\min})^{-s_1} \right] e^{-Ns_1 \tau_{(\mbox{\tiny\boldmath$r$}, P^{\mathcal{F}}_{\min} ,\mathcal{S})}(\mbox{\scriptsize\boldmath$y$}) } \right],
\end{eqnarray}
and upper bound $P_{i[(\tilde{\mbox{\scriptsize\boldmath$r$}},\tilde{\mathcal{F}}),(\mbox{\scriptsize\boldmath$r$},\mathcal{F}),  \tilde{P}_{Y|\mbox{\tiny\boldmath$X$}}, \mathcal{S}]} $ for any $s_2 >0, 0 < \tilde{\rho} \le 1$ by,
\begin{eqnarray}
\label{CCBound4P_i}
P_{i[(\tilde{\mbox{\scriptsize\boldmath$r$}},\tilde{\mathcal{F}}),(\mbox{\scriptsize\boldmath$r$},\mathcal{F}), \tilde{P}_{Y|\mbox{\tiny\boldmath$X$}}, \mathcal{S}]}
&\le&  \sum_{\mbox{\scriptsize\boldmath$y$}} E_{\mbox{\scriptsize\boldmath$\theta$}_{\mathcal{S}}} \left[  E_{\mbox{\scriptsize\boldmath$\theta$}_{\bar{\mathcal{S}}}} \left[ P(\mbox{\boldmath$y$}|\mbox{\boldmath$x$}_{(\tilde{\mbox{\scriptsize\boldmath$w$}},\tilde{\mbox{\scriptsize\boldmath$r$}})},\tilde{P}_{Y|\mbox{\scriptsize\boldmath$X$}}) \right] \right. \nonumber\\
&& \left. \times E_{\mbox{\scriptsize\boldmath$\theta$}_{\bar{\mathcal{S}}}} \left\{ \left[ P(\mbox{\boldmath$y$}|\mbox{\boldmath$x$}_{(\mbox{\scriptsize\boldmath$w$},\mbox{\scriptsize\boldmath$r$})},P^{\mathcal{F}}_{\min})^{\frac{s_2}{\tilde{\rho}}} \right] \right\}^{\tilde{\rho}} e^{Ns_2 \tau_{(\mbox{\tiny\boldmath$r$}, P^{\mathcal{F}}_{\min},\mathcal{S})}(\mbox{\scriptsize\boldmath$y$}) } e^{N\tilde{\rho}\sum_{k\notin\mathcal{S}}r_k} \right],            \nonumber\\
&\le & \max_{(\mbox{\scriptsize\boldmath$r$}', \mathcal{F}') \notin \mathcal{R},\mbox{\scriptsize\boldmath$r$}'_{\mathcal{S}} = \mbox{\scriptsize\boldmath$r$}_{\mathcal{S}}} \sum_{\mbox{\scriptsize\boldmath$y$}} E_{\mbox{\scriptsize\boldmath$\theta$}_{\mathcal{S}}} \left[  E_{\mbox{\scriptsize\boldmath$\theta$}_{\bar{\mathcal{S}}}} \left[ P(\mbox{\boldmath$y$}|\mbox{\boldmath$x$}_{(\mbox{\scriptsize\boldmath$w$}',\mbox{\scriptsize\boldmath$r$}')},P^{\mathcal{F}'}_{\max}) \right] \right. \nonumber\\
&&\left. \times E_{\mbox{\scriptsize\boldmath$\theta$}_{\bar{\mathcal{S}}}} \left\{ \left[ P(\mbox{\boldmath$y$}|\mbox{\boldmath$x$}_{(\mbox{\scriptsize\boldmath$w$},\mbox{\scriptsize\boldmath$r$})},P^{\mathcal{F}}_{\min})^{\frac{s_2}{\tilde{\rho}}} \right] \right\}^{\tilde{\rho}} e^{Ns_2 \tau_{(\mbox{\tiny\boldmath$r$}, P^{\mathcal{F}}_{\min},\mathcal{S})}(\mbox{\scriptsize\boldmath$y$}) } e^{N\tilde{\rho}\sum_{k\notin\mathcal{S}}r_k} \right].
\end{eqnarray}
Note that the upper bound given in (\ref{CCBound4P_t}) is not a function of $P_{Y|\mbox{\scriptsize\boldmath$X$}}$. Similarly, the bound in (\ref{CCBound4P_i}) is not a function of $\tilde{P}_{Y|\mbox{\scriptsize\boldmath$X$}}$.

Optimization of the typicality threshold $\tau_{(\mbox{\tiny\boldmath$r$}, P^{\mathcal{F}}_{\min},\mathcal{S})}$ can be carried out using the similar technique as introduced in (\ref{TauEquality})-(\ref{TauExpression}) in Appendix \ref{AppendixTheoremMC}. By substituting the optimal $\tau_{(\mbox{\tiny\boldmath$r$}, P^{\mathcal{F}}_{\min},\mathcal{S})}$ into (\ref{CCBound4P_t}) and (\ref{CCBound4P_i}), we get
\begin{eqnarray}
\label{BoundPiPtCC}
P_{t[\mbox{\scriptsize\boldmath$r$},\mathcal{F},P_{Y|\mbox{\tiny\boldmath$X$}},\mathcal{S}]}, P_{i[(\tilde{\mbox{\scriptsize\boldmath$r$}},\tilde{\mathcal{F}}),(\mbox{\scriptsize\boldmath$r$},\mathcal{F}),\tilde{P}_{Y|\mbox{\tiny\boldmath$X$}},\mathcal{S}]} \le \max_{(\mbox{\scriptsize\boldmath$r$}', \mathcal{F}') \notin \mathcal{R},\mbox{\scriptsize\boldmath$r$}'_{\mathcal{S}} = \mbox{\scriptsize\boldmath$r$}_{\mathcal{S}}} \exp \left\{ -NE_i(\mathcal{S}, \mbox{\boldmath$r$},\mbox{\boldmath $r$}', \mathcal{F}, \mathcal{F}' ) \right\},
\end{eqnarray}
where $E_i(\mathcal{S}, \mbox{\boldmath$r$},\mbox{\boldmath $r$}', \mathcal{F}, \mathcal{F}' ) $ is given in (\ref{EmEiCC}).

Combining (\ref{CCBound4P_m}), (\ref{BoundPiPtCC}) and (\ref{ProofPesCC}), we obtain
\begin{eqnarray}
\label{CCPs1}
&& P_{es} \le \max \left\{ \max_{(\mbox{\scriptsize\boldmath$r$},P_{Y|\mbox{\tiny\boldmath$X$}}): P_{Y|\mbox{\tiny\boldmath$X$}} \in \mathcal{F}, (\mbox{\scriptsize\boldmath$r$},\mathcal{F}) \in \mathcal{R}} \sum_{\mathcal{S} \subset \{1,\cdots,K\}} \left[ \max_{(\mbox{\scriptsize\boldmath$r$}', \mathcal{F}') \notin \mathcal{R},\mbox{\scriptsize\boldmath$r$}'_{\mathcal{S}} = \mbox{\scriptsize\boldmath$r$}_{\mathcal{S}}} \exp \left\{ -NE_i(\mathcal{S}, \mbox{\boldmath$r$},\mbox{\boldmath $r$}', \mathcal{F}, \mathcal{F}' ) \right\} \right. \right.\nonumber\\
&& \qquad \left. \qquad\qquad\qquad\qquad\qquad\qquad + \sum_{(\tilde{\mbox{\scriptsize\boldmath$r$}},\tilde{\mathcal{F}})  \in \mathcal{R},\tilde{\mbox{\scriptsize\boldmath$r$}}_{\mathcal{S}} = \mbox{\scriptsize\boldmath$r$}_{\mathcal{S}}} \exp \left\{ -NE_m(\mathcal{S}, \mbox{\boldmath$r$},\tilde{\mbox{\boldmath $r$}}, \mathcal{F}, \tilde{\mathcal{F}} ) \right\} \right] , \nonumber\\
&& \max_{(\tilde{\mbox{\scriptsize\boldmath$r$}},\tilde{P}_{Y|\mbox{\tiny\boldmath$X$}} ): \tilde{P}_{Y|\mbox{\tiny\boldmath$X$}} \in \mathcal{F}, (\tilde{\mbox{\scriptsize\boldmath$r$}},\tilde{\mathcal{F}})  \notin \mathcal{R}} \sum_{\mathcal{S} \subset \{1,\cdots,K\}}
\left. \left[ \sum_{(\mbox{\scriptsize\boldmath$r$},\mathcal{F}) \in \mathcal{R},\mbox{\scriptsize\boldmath$r$}_{\mathcal{S}} = \tilde{\mbox{\scriptsize\boldmath$r$}}_{\mathcal{S}}} \max_{(\mbox{\scriptsize\boldmath$r$}', \mathcal{F}') \notin \mathcal{R},\mbox{\scriptsize\boldmath$r$}'_{\mathcal{S}} = \mbox{\scriptsize\boldmath$r$}_{\mathcal{S}}} \exp \left\{ -NE_i(\mathcal{S}, \mbox{\boldmath$r$},\mbox{\boldmath $r$}', \mathcal{F}, \mathcal{F}' ) \right\} \right] \right\}.
\end{eqnarray}
Since the upper bounds given in (\ref{CCBound4P_m}) and (\ref{BoundPiPtCC}) are not functions of individual channels (but functions of channel classes), the right hand side of (\ref{CCPs1}) can be simplified to the right hand side of (\ref{PesCC}).




\begin{thebibliography}{1}

\bibitem{ref Luo09}
J. Luo and A. Ephremides, {\em A New Channel Coding Approach for Random Access with Bursty Traffic}, to appear 
in IEEE Trans. on Inform. Theory.

\bibitem{ref Massey85}
J. Massey and P. Mathys, {\em The Collision Channel Without Feedback,}
IEEE Trans. on Inform. Theory, Vol. IT-31, pp. 192-204, Mar. 1985.

\bibitem{ref WangJournal}
Z. Wang and J. Luo, {\em Error Performance of Channel Coding in Random Access Communication}, submitted to IEEE Trans. on Inform. Theory. http://arxiv.org/abs/1010.0642.

\bibitem{ref Lapidoth98}
A. Lapidoth and P. Narayan, {\em Reliable Communication under Channel Uncertainty,}
IEEE Trans. Inform. Theory, Vol. 44, pp. 2148-177, Oct. 1998.

\bibitem{ref Blackwell59}
D. Blackwell, L. Breiman, and A. Thomasian, {\em The capacity of A Class of Channels,}
Ann. Math. Statist., Vol. 30, pp. 1229-241, 1959.

\bibitem{ref Wolfowitz59}
J. Wolfowitz, {\em Simultaneous Channels,}
Arch. Rational Mech. and Anal., Vol. 4, pp. 371-86, 1959.

\bibitem{ref Cover05}
T. Cover and J. Thomas, {\em Elements of Information Theory,}
2nd Ed.,Wiley Interscience, 2005.

\bibitem{ref Csiszar81}
I. Csisz\'{a}r and J. Korner, {\em Information Theory: Coding Theorems for Discrete Memoryless Systems}, New York: Academic, 1981.

\bibitem{ref Bertsekas92}
D. Bertsekas and R. Gallager, {\em Data Network,}
2nd Ed., Prentice Hall, 1992.

\bibitem{ref Shamai07}
S. Shamai, I. Teletar, and S. Verd\'{u}, {\em Fountain Capacity,}
IEEE Trans. Inform. Theory, Vol. 53, pp. 4372-4376, Nov. 2007.


\end{thebibliography}
%




\end{document}